\documentclass[12pt]{article}
\usepackage{amsmath}
\usepackage{graphicx}
\usepackage{enumerate}
\usepackage{natbib}
\usepackage{url} 

\usepackage{amssymb}
\usepackage[dvips]{epsfig}
\usepackage{dsfont}
\usepackage[usenames,dvipsnames]{xcolor}
\usepackage{rotating}
\usepackage[breaklinks]{hyperref}
\usepackage{makecell}
\usepackage[percent]{overpic}
\usepackage[font={stretch=1.0,footnotesize}]{caption}
\usepackage{mathtools}
\usepackage{algorithmicx}
\usepackage[]{algpseudocode}
\usepackage[section]{algorithm}
\usepackage{diagbox}
\usepackage{wrapfig}
\input epsf
\usepackage{fontenc}
\usepackage{setspace}
\usepackage{bm}
\usepackage{slashbox}
\usepackage{multirow}
\usepackage{eurosym}
\usepackage{bbm}
\usepackage{longtable}
\usepackage{subcaption}
\usepackage[flushleft]{threeparttable}

\epsfverbosetrue

\usepackage{amsthm}

\newtheorem{corollary}{Corollary}
\newtheorem{assumption}{Assumption}
\newtheorem{proposition}{Proposition}

\theoremstyle{definition}
\newtheorem{remark}{Remark}

\usepackage{titlesec}

\usepackage{scalerel,stackengine}
\stackMath
\newcommand\reallywidehat[1]{%
\savestack{\tmpbox}{\stretchto{%
  \scaleto{%
    \scalerel*[\widthof{\ensuremath{#1}}]{\kern-.6pt\bigwedge\kern-.6pt}%
    {\rule[-\textheight/2]{1ex}{\textheight}}
  }{\textheight}%
}{0.5ex}}%
\stackon[1pt]{#1}{\tmpbox}%
}

\def \dsE {\text{$\mathds{E}$}}
\def \dsR {\text{$\mathds{R}$}}

\DeclareMathOperator{\rank}{rk}

\def \avec {\text{\boldmath$a$}}

\def \dvec {\text{\boldmath$d$}}    
    
\def \fvec {\text{\boldmath$f$}}

\def \wvec {\text{\boldmath$w$}}    
\def \xvec {\text{\boldmath$x$}}    
\def \yvec {\text{\boldmath$y$}}    \def \mY {\text{\boldmath$Y$}}

\def \betavec         {\text{\boldmath$\beta$}}

\def \varepsilonvec   {\text{\boldmath$\varepsilon$}}
\def \zetavec         {\text{\boldmath$\zeta$}}

\def \thetavec        {\text{\boldmath$\theta$}}
\def \varthetavec     {\text{\boldmath$\vartheta$}}

\def \lambdavec       {\text{\boldmath$\lambda$}}
\def \muvec           {\text{\boldmath$\mu$}}

\def \xivec           {\text{\boldmath$\xi$}}

\def \rhovec          {\text{\boldmath$\rho$}}

\def \tauvec          {\text{\boldmath$\tau$}}

\def \phivec          {\text{\boldmath$\phi$}}

\def \nullvec {\mathbf{0}}

\newcommand{\blind}{1}

\usepackage{paralist}
\renewenvironment{itemize}[1]{\begin{compactitem}#1}{\end{compactitem}}
\renewenvironment{enumerate}[1]{\begin{compactenum}#1}{\end{compactenum}}
\begin{document}

\def\spacingset#1{\renewcommand{\baselinestretch}%
{#1}\small\normalsize} \spacingset{1}


\if1\blind
{
  \title{\bf Scalable Estimation for Structured Additive  Distributional Regression Through Variational Inference}
  \author{Jana Kleinemeier and Nadja Klein\thanks{
    Nadja Klein was supported by the Deutsche Forschungsgemeinschaft (DFG, German Research Foundation) through the Emmy Noether grant KL 3037/1-1. }\\
    \\
    Algonaut GmbH and\\
    Chair of Uncertainty Quantification and Statistical Learning,\\  
    Department of Statistics, Technische Universit\"at Dortmund\\}
  \maketitle
} \fi

\if0\blind
{
  \bigskip
  \bigskip
  \bigskip
  \begin{center}
    {\LARGE\bf Scalable Estimation for Structured Additive\\  
    Distributional Regression Through\\ \vspace{0.3cm}
    
    Variational Inference}
\end{center}
  \medskip
}\fi

\thispagestyle{empty}

\bigskip
\begin{abstract}

Structured additive distributional regression models offer a versatile framework for estimating complete conditional distributions by relating all parameters of a parametric distribution to  covariates. Although these models efficiently leverage information in vast and intricate data sets, they often result in highly-parameterized models with many unknowns. Standard estimation methods, like Bayesian approaches based on Markov chain Monte Carlo methods, face challenges in estimating these models due to their complexity and costliness. To overcome these issues, we suggest a fast and scalable alternative based on variational inference. Our approach combines a parsimonious parametric approximation for the  posteriors of regression coefficients, with the exact conditional posterior for hyperparameters. For optimization, we use a stochastic gradient ascent method combined with  an efficient strategy to reduce the variance of estimators. We provide theoretical properties and investigate global and local annealing to enhance robustness, particularly against data outliers. Our implementation is very general, allowing  us to include various functional effects like penalized splines or complex tensor product interactions. In a simulation study, we demonstrate the efficacy of our approach in terms of accuracy and computation time. Lastly, we present two real examples illustrating the modeling of infectious COVID-19 outbreaks and outlier detection in brain activity.
\end{abstract}

\noindent%
\vfill
\setcounter{page}{0}
\newpage
\spacingset{1.9} 
\setlength{\abovedisplayskip}{0.0cm}
\setlength{\belowdisplayskip}{0.0cm}

\section{Introduction}\label{sec:intro}
Traditionally, conditional mean regression models have received the most attention in statistical modeling. However, over the last decade, there has been a significant paradigm shift towards distributional regression models. This shift is motivated by the recognition that, in many fields, modeling the mean is not the sole or primary concern. Instead, there is a growing interest in probabilistic modeling, where quantities of interest extend beyond the mean to include extreme observations, conditional quantiles, and exceedance probabilities for specific response thresholds. Distributional regression encompasses various methods, such as quantile regression for estimating different quantiles of the conditional distribution \citep{Koe2005}, distribution regression \citep{ForPer1995}, isotonic distributional regression \citep{HenZieGne2021}, Gaussian process regression for modeling non-linear relationships and uncertainty estimation, or mixture models \citep{FruCelRob2019}; see e.g.~\citet{Kle2023} for a recent review on distributional regression.

However, achieving realistic models often requires more than just  flexibility in capturing the response distribution. Real-world data exhibits intricate and non-linear relationships between covariates and the outcome of interest. This recognition has led to a demand for greater modeling flexibility that moves beyond simple linear associations. Non-parametric smoothing functions like penalized splines or complex tensor product interactions can be used for interpretable data-driven flexibility determination.

For example, in the context of public health, like predicting daily COVID-19 infections as in our first illustration, the need for a flexible modeling approach is paramount. Daily infection counts often follow complex, non-Gaussian distributions with dynamic patterns that conventional models struggle to capture. Spatial information, time-dependent interventions, and non-linear effects are critical factors in understanding the virus' spread. A rigid modeling framework relying solely on parametric assumptions and linear relationships may fall short in capturing the multifaceted dynamics at play. 

Thus, realistic models often necessitate a complex approach, combining flexibility in capturing the response distribution and covariate effects, particularly when dealing with large data sets. In this paper, we focus on the class of generalized additive models for generalized additive models for location, scale and shape (GAMLSS) introduced by \citet{RigRobSta2005}, also known as structured additive distributional regression \citep[SADR;][]{KleKneLanSoh2015}. These models extend generalized additive models \citep[GAMs;][]{Woo2006} to arbitrary response distributions, where each distributional parameter can be related to covariates through a semiparametric predictor. Estimation of such models has so far been proposed in the penalized maximum likelihood \citep{RigRobSta2005}, statistical boosting \citep{MayFenKneSch2011}  and Bayesian \citep{KleKneLan2015} frameworks. 
The estimation of such highly parameterized models is often very time consuming and in the case of Bayesian estimation based on Markov chain Monte Carlo methods as in \citet{KleKneLanSoh2015} can lead to very slow convergence. While the approach of \citep{UmlSeiWetSimLanKle2023}, which proposes a novel backfitting algorithm based on stochastic gradient descent that can handle large data sets and performs automatic variable selection and smoothing parameter estimation, scales well to large-scale data, it does not allow for uncertainty quantification and inference. 
To address both, scalability to complex GAMLSS and access to uncertainty estimates of any quantity of interest derived from the conditional distribution,  we develop an approximate Bayesian approach to SADR models based on variational inference (VI) in which the posterior distribution is approximated by a tractable distribution \citep{Ble2017}. We use the Gaussian family as a variational approximation (VA) which is a  popular choice and often provides very accurate approximations to  posterior means/modes. Traditionally,  a diagonal covariance matrix of the VA is assumed. This however, implies that no posterior dependence between parameters exists \citep{KucRajGelBle2015} which  can be too restrictive in many situations. Instead, we employ a flexible factor covariance structure \citep{OngNotSmi2018} which allows for correlation between the random variables. Our approach combines this parsimonious parametric approximation for the  posteriors of regression coefficients, with the exact conditional posterior for hyperparameters. In order to minimize the  Kullback-Leibler divergence between the true posterior and the approximation we use a stochastic gradient ascent method combined with  an efficient way to reduce the variance of the estimators \citep{KinWel2014}.

Similar to \citet{KucRajGelBle2015}, who propose an automatic differentiation VI algorithm called \texttt{advi}, we build a model agnostic estimation framework, where a large variety of different functional effects and distributions can be used. 

Though VI offers scalability, it frequently encounters challenges associated with local optima. To address this issue, we extend the newly introduced estimation approach by developing complementary robust variants. These include the capacity to subsample the likelihood at each  step and two annealing methods inspired by \citet{ManMciAbrRanBle2016}, called global annealing and local annealing or Bayesian data re-weighting. Global annealing involves introducing a global temperature parameter, which facilitates exploration of the loss function in the initial stages of optimization. In Bayesian data re-weighting, reverse temperatures are assigned to each data point, allowing for the down-weighting of outliers and thus a dynamic annealing.

Overall, our paper makes the following important contributions to SADR models:
\begin{itemize}
 \item the introduction of an approximate, scalable estimation approach  that allows for uncertainty quantification and inference,
 \item an estimation framework agnostic to the specific model specification,
 \item the introduction of global annealing and subsampling approaches, and
 \item  the introduction of Bayesian data re-weighting as  robust variants.
\end{itemize}
The rest of this paper is structured as follows: Section~\ref{sec:SADR} summarises the specification of Bayesian SADR models.  Sections~\ref{sec:VB} and \ref{sec:extend} contain details and properties of our scalable posterior estimation approach through VI, software and implementation but also extensions to robust fitting. Sections~\ref{sec:simul} and \ref{sec:appl} evaluate the performance of our approach in simulations and two diverse applications. The final Section~\ref{sec:discussion}  concludes. A Supplement gives further details on computation, simulations and applications.

\setlength{\abovedisplayskip}{0.0cm}
\setlength{\belowdisplayskip}{0.0cm}

\section{\mbox{Bayesian Structured Additive Distributional Regression}}\label{sec:SADR}

Let $\lbrace (\yvec _{i}, \xvec _{i})\rbrace_{i=1,\ldots,n}$,
denote $n$ conditionally  independent data points of observations on a response variable
$\mY \in \mathcal{Y}\subseteq \dsR ^{p}$, $p\geq 1$ and $\xvec $ the covariate
vector comprising different types of covariate information such as discrete
and continuous covariates or spatial information. We focus on $p=1$ but generalizations to multivariate responses would conceptually be straightforward. In SADR it is assumed that
the conditional distribution of $\yvec_{i}$ given $\xvec_{i}$ is specified
via a $K$-parametric distribution with density
%
$p_Y(\yvec _{i}\,|\,\vartheta _{i1},\ldots ,\vartheta _{iK}),$
%
where $\varthetavec _{i}=(\vartheta _{i1},\ldots ,\vartheta _{iK})^\top$ is
a collection of $K$ observation specific scalar distributional parameters
$\vartheta _{ik}$, $k=1,\ldots ,K$. Various simpler models, such as generalized
additive or survival models are included as special cases. While in traditional mean regression models with $p_Y(\cdot\, |\,\varthetavec)$ from the exponential family the focus is on modelling $\vartheta _{i1}=\dsE (y_{i})$, and all other
$K-1$ parameters are treated as fixed or nuisance parameters, in SADR, each of the distributional parameters
$\vartheta _{ik}\equiv\vartheta _{ik}(\xvec_i)$ is related to regression effects. To maintain potential restrictions on the parameter spaces we write $h_{k}(
\eta _{ik})=\vartheta _{ik}$ and
$\eta _{ik} = h^{-1}_{k}(\vartheta _{ik})$, where $h_k$ are one-to-one response functions with inverses $h^{-1}_{k}$. 

\subsection{Semiparametric predictors}

Each of the $K$ predictors is of structured additive form, i.e.~$
 \eta_{ik} =   \sum_{j=1}^{J_k}f_{j,k}(\xvec_i),
$
where the effects $f_{j,k}(\xvec_i)$ represent   flexible functions depending on (different subsets of) the covariate vector $\xvec_i$ that are modelled through appropriate basis function representations
\[
 f_{j,k}(\xvec_i) = \sum_{d=1}^{D_{j,k}}\beta_{j,k,d}B_{j,k,d}(\xvec_i),
\]
where $B_{l,k,d}(\xvec_i)$, $d=1,\ldots,D_{j,k}$ are appropriate basis functions and $\betavec_{j,k}=(\beta_{j,k,1},\ldots,\beta_{j,k,D_{j,k}})^\top$ are the unknown vectors of basis coefficients. Due to the linear basis representation, the vector of function evaluations $\fvec_{j,k}=(f(\xvec_{j,k,1}),\ldots,f(\xvec_{j,k,n}))^\top$ can now be written as $\fvec_{j,k}=B_{j,k}\betavec_{j,k}$ where $B_{j,k}$ is the ($n\times D_{j,k}$) design matrix arising from the evaluation of the basis functions $B_{j,k,d}(\xvec_i)$, $d=1,\ldots,D_{j,k}$, at the observed  $\xvec_1,\ldots,\xvec_n$.

{We briefly discuss some of the components $f_{j,k}$ used later:
\begin{itemize}
\item For linear effects of continuous covariates, the columns of the design matrix $B_{j,k}$ are  equal to the original covariates. For binary/categorical covariates, the basis functions represent the chosen coding, e.g.~dummy or effect coding and the design matrix then consists of the resulting dummy or effect coding columns. 
\item For a  nonlinear effect of a univariate continuous covariate $x$  splines are common choices, such as Bayesian P-splines, smoothing splines and cyclic cubic splines. 
\item For nonlinear continuous interaction effects we use  tensor product P-splines or thin plate splines depending on the use case. 
\item Spatial effects for a discrete set of geographical regions are modelled via Gaussian Markov random fields \citep[GMRFs;][]{RueHel2005} where the design matrix has entries $(i,s)$ equal to one if observation $i$ is located in region $s$ and zero otherwise.
\end{itemize}}
If not specified otherwise we use the default settings of the \texttt{mgcv} package which take an approximation of the smoothing splines and thin plate splines as defaults for univariate and multivariate nonlinear effects. 
\subsection{Prior specifications}
The Bayesian treatment of SADR models is completed by making appropriate prior assumptions for the regression coefficients $\betavec_{j,k}$.  Since for many  types of effects the vector of basis coefficients $\betavec_{j,k}$ is of relatively high dimension, it is  useful to enforce specific properties such as smoothness or shrinkage. In a Bayesian formulation, this can be facilitated by assuming constrained multivariate Gaussian priors
\[	
 p(\betavec_{j,k}|\tauvec_{j,k}^2)\propto\exp\left(-\frac{1}{2}\betavec_{j,k}^\top K_{j,k}(\tauvec_{j,k}^2)\betavec_{j,k}\right)\mathds{1}\left\lbrack A_{j,k}\betavec_{j,k}=\nullvec\right\rbrack,
\]
where $K_{j,k}\equiv K_{j,k}(\tauvec_{j,k}^2)\in\dsR^{D_{j,k}\times D_{j,k}}$ denotes the prior precision matrix implementing the desired smoothness properties and the degree of smoothness is controlled by further hyperparameters $\tauvec_{j,k}^2$. The indicator function $\mathds{1}[A_{j,k}\betavec_{j,k}= \nullvec]$ is included to enforce linear constraints on the regression coefficients via the constraint matrix $A_{j,k}$. The latter is typically used to remove identifiability issues of the additive predictor (e.g.~by centering the additive components of the predictor) but can also be used to remove the partial impropriety from the prior that comes from a potential rank deficiency of  $K_{j,k}$, when $\rank(K_{j,k})=\kappa_{j,k}< D_{j,k}$.

For the examples of effect types above we make the following choices for $K_{j,k}$:
\begin{itemize}
    \item  For linear effects  flat improper priors (with $K_{j,k}=0$) are common. An alternative are informative Gaussian priors (e.g.~a Bayesian ridge regression prior with $K_{j,k}=\tau_{j,k}^2 I)$ that enforce shrinkage of the effects towards zero.
    \item For nonlinear effects we follow \citet{UmlKleZei2018} and match the prior depending on the basis function.
    \item For discrete spatial effects we use the precision matrix induced by an adjacency matrix encoding the neighbourhood relation between the regions~\citep{RueHel2005}.
\end{itemize}

Different hyperpriors for $\tauvec_{j,k}^2$ have been proposed in the literature. On the one hand, we consider the standard conjugate case of inverse gamma priors, i.e.~$\tau_{j,k}^2\sim IG(a_{j,k},b_{j,k})$, where we set  the hyperparameters to $a_{j,k}=b_{j,k}=0.001$ as a default following \citet{KleKneLan2015}. 
On the other hand, we use the scale-dependent priors of \citet{KleKne2016}, that is Weibull priors with shape equal to 0.5 and the scale parameter equal to 0.0088 as default. In Section~\ref{sec:simul}, we compare the two different prior distributions in terms of estimation performance.

To conclude we define the set of all unknown model parameters given the design matrices $B_{j,k}$ and prior precision matrices $K_{j,k}$  by 
$
\thetavec=(\betavec^\top,(\log(\tauvec^2))^\top)^\top=(\betavec^\top,(\tilde\tauvec^2)^\top)^\top\in\dsR^{p_{\theta}}$,  $p_{\theta}=\dim(\betavec)+\dim(\tilde\tauvec^2)=p_{\beta}+p_{\tilde\tau}$,
where $\betavec=(\betavec_1,\ldots,\betavec_K)^\top,\tilde\tauvec^2=(\tilde\tauvec_1^2,\ldots,\tilde\tauvec_K^2)^\top$, $\betavec_k=(\betavec_{1,k},\ldots,\betavec_{J_k,k})^\top$, $\tilde\tauvec_k^2=(\tilde\tauvec_{1,k}^2,\ldots\tilde,\tauvec_{J_k,k})^\top.$ The transformation of the marginal variances to the real line is done for convenience to make the respective posterior distributions closer to normality, see Section~\ref{subsec:VA} for details.  In what follows, we use $p_Y(\cdot\,|\,\xvec,\thetavec)\equiv p_Y(\cdot\,|\,\xvec,\betavec)\equiv p_Y(\cdot\,|\,\varthetavec(\xvec))$ interchangeable for the likelihood.

\subsection{Model choice and variable selection}

To select a reasonable SADR model, we consider normalized quantile residuals \citep{DunSmyXXX} as a graphical device in a first step. Comparing these residuals visually  is helpful to select from a set of  candidate response distributions that yield appropriate overall fits.

In a second step, we use the  Watanabe-Akaike information criterion \citep[WAIC;][]{watanabe2010} not only to confirm the best fitting distribution but also the predictor specifications. The WAIC can be seen as an approximation to computationally expensive cross validation (CV) and it is conveniently computed from $s=1, \ldots, S$ posterior samples.  It overcomes certain limitations of the deviance information criterion (DIC) such as its dependence on the posterior mean as a specific point estimate or the potential of observing negative effective parameter counts and the assumption of posterior normality for the posterior. For a number of $S$ MCMC samples from the posterior, the WAIC is given by
$
\text{WAIC} = (-2l_{\text{WAIC}} + 2p_{\text{WAIC}}),
$
where $l_{\text{WAIC}} = \sum_{i=1}^n \log \left( \frac{1}{S} \sum_{s=1}^S p_Y(\yvec_i\,|\,\xvec_i, \thetavec^{[s]})\right)
$ is the log pointwise predictive density and
$ p_{\text{WAIC}} =   \sum_{i=1}^n \frac{1}{S-1}\sum_{s=1}^S\left\lbrace\log(p_Y(\yvec_i\,|\,\xvec_i,\thetavec^{[s]})) - \frac{1}{S}\sum_{s=1}^S\log(p_Y(\yvec_i\mid \xvec_i,\thetavec^{[s]}))\right\rbrace^2$ is the effective number of parameters \citep{VehAkiGel2017}. 

When primary interest is in predictive performance, it is advisable to conduct e.g.~CV combined with the evaluation of proper scoring rules \citep{GneTilRaf2007}. 

\setlength{\abovedisplayskip}{0.0cm}
\setlength{\belowdisplayskip}{0.0cm}

\section{Posterior Estimation via Variational Bayes}\label{sec:VB}

To perform Bayesian inference for SADR we   consider VI methods, in which a member $q_\lambda(\thetavec)$ of some parametric family of densities with so-called variational parameters $\lambdavec$ is used to approximate
the target posterior $p(\thetavec\,|\,\yvec)\propto g(\thetavec)=p_Y(\yvec\,|\,\xvec,\thetavec)p(\thetavec)=p_Y(\yvec\,|\,\betavec)p(\betavec\,|\,\tilde{\tauvec}^2)p(\tilde{\tauvec}^2)$. We first provide a short overview of VI in general before we outline details on how VI can be successfully used to perform scalable posterior estimation in Bayesian SADR models.

\subsection{Key idea of variational inference}

Approximate Bayesian inference through VI defines an optimization problem, where the variational parameters $\lambdavec$ are tailored towards a member $q_\lambda(\thetavec)$  that is ``close to'' $p(\thetavec\,|\,\yvec)$. Proximity between  $q_\lambda(\thetavec)$  and $p(\thetavec\,|\,\yvec)$ is given by a divergence measure. For the latter, the Kullback–Leibler (KL) divergence $\mbox{KL}(q_\lambda(\thetavec)\,||\,p(\thetavec\,|\,\yvec))$ is typically employed, and it is straightforward to show that minimizing the KL divergence is  equivalent to maximizing the variational lower bound \citep[also called the evidence lower bound, or ``ELBO''; see e.g.~][]{OrmWan2010,Ble2017}
given by
\begin{align*}
		\mathcal{L}(\boldsymbol{\lambda}) =  \int q_\lambda(\thetavec)\log\left(\frac{p(\yvec\,|\,\thetavec)p(\thetavec)}{q_\lambda(\thetavec)}\right)d\thetavec.
\end{align*}
The ELBO takes the form of an intractable integral. Yet, recognizing that it can be written as an expectation with respect to $q_\lambda$ as
\begin{equation}\label{eq:ELBO}
\mathcal{L}(\lambdavec)=\dsE_{q_\lambda}\lbrack \log g(\thetavec)-\log q_\lambda(\thetavec)\rbrack,
\end{equation}
where $g(\thetavec)=p_Y(\yvec\mid\thetavec)p(\thetavec)$, it can be optimized using stochastic gradient ascent methods \citep[SGA;][]{Bot2010}. Given an initial value $\lambdavec^{(0)}$, SGA sequentially optimizes the ELBO through
\[\boldsymbol{\lambda}^{(t+1)}=\boldsymbol{\lambda}^{(t)}+\boldsymbol{\rho}^{(t)} \circ \nabla_\lambda\widehat{\mathcal{L}(\boldsymbol{\lambda}^{(t)})}, \quad t=1,\ldots,\]
where $\rhovec^{(t)}$ is a vector of step sizes, $\circ$ denotes the element-wise product of two vectors and $\nabla_\lambda\widehat{\mathcal{L}(\boldsymbol{\lambda}^{(t)})}$ is an unbiased estimate of the gradient of $\mathcal{L}(\lambdavec)$ at $\lambdavec=\lambdavec^{(t)}$. For appropriate step size choices this will converge to a local
optimum of $\mathcal{L}(\lambdavec)$ \citep{RobMon1951}. Adaptive step size choices are often used in practice, and we use the  automatic ADADELTA
method of \citep{Zei2012} which has proven to work well in the context of smoothing models.

In principle, unbiased estimates of the gradient can be obtained by directly differentiating \eqref{eq:ELBO} with respect to $\lambdavec$ and by approximating the expectation through simulation from $q_\lambda$. However, variance reduction methods for the gradient estimation are often needed  for fast
convergence and stability. Here, we use the  ``re-parameterization trick'' \citep{KinWel2014}, in which it is assumed that $\thetavec$ can be generated from $q_\lambda$ by first generating $\zetavec$ from  density $f_\zeta$ not depending on $\lambdavec$ and then applying a deterministic transformation $\thetavec=t(\zetavec,\lambdavec)$ to obtain $\thetavec$. In this case, \eqref{eq:ELBO} can be written as
\begin{equation}\label{eq:ELBO2}
\mathcal{L}(\lambdavec)=\dsE_{f_\zeta}\lbrack \log g(t(\zetavec,\lambdavec))-\log q_\lambda(t(\zetavec,\lambdavec))\rbrack,
\end{equation}
and differentiating under the integral sign in  \eqref{eq:ELBO2} yields the ``re-parameterization gradient''
\begin{equation}\label{eq:gradELBO2}\begin{aligned}
\nabla_\lambda\mathcal{L}(\lambdavec)&=\dsE_{f_\zeta}\lbrack \nabla_\lambda\lbrace \log g(t(\zetavec,\lambdavec))-\log q_\lambda(t(\zetavec,\lambdavec))\rbrace\rbrack\\&=\dsE_{f_\zeta}\left\lbrack \frac{\partial t(\zetavec,\lambdavec)^\top}{\partial\lambdavec}\nabla_\theta\lbrace \log g(t(\zetavec,\lambdavec))-\log q_\lambda(t(\zetavec,\lambdavec))\rbrace\right\rbrack.
\end{aligned}\end{equation}
Note that \eqref{eq:gradELBO2} employs gradient information from the log-posterior by moving $\lambdavec$ inside $g(\cdot)$ and allows fast sampling from $f_\zeta$. We detail our choice for $t$ in Section~\ref{subsubsec:GVA}. In practice, for a well-chosen VA only a few draws from $f_\zeta$ are sufficient for the SGA to converge fast. We investigate this in more detail empirically in Section~\ref{sec:simul} where we find that often even one draw suffices.

\subsection{Variational approximations for SADR}\label{subsec:VA}

We consider two VAs  for $\thetavec$. Our first choice is to approximate $p(\thetavec\,|\,\yvec)$ by a tractable variational density of a fixed form density $q_\lambda(\thetavec)=q_\lambda(\betavec,\tilde\tauvec^2)$. In doing so, proper calibration of $q_\lambda$ has the potential to be more efficient than MCMC sampling as it does not require partitioning $\thetavec$ or possibly inefficient Metropolis-Hastings (MH) steps as in  MCMC. This can be particularly appealing  for the high-dimensional subvectors $\betavec_{j,k}$ in SADR with intractable and highly dependent conditional posterior distributions.  Recognizing however, that when standard inverse gamma priors for the variances $\tauvec_{j,k}^2$ are employed,  respective full conditional distributions $p(\tau_{j,k}^2\,|\,\betavec\setminus\tau_{j,k}^2,\yvec)$ are of closed form, it is possible to define are more accurate VA independent of the specific fixed form density considered which does not need a logarithmic transformation of $\tauvec^2$. In this case, we propose to replace $q_\lambda(\thetavec)$ by
\begin{align}\label{eq:va2}
q_\lambda(\betavec,\tauvec^2)=q_\lambda^0(\betavec)\prod_{k=1}^K\prod_{j=1}^{J_k}  p(\tau_{j,k}^2\,|\,\betavec\setminus\tau_{j,k}^2,\yvec)
 \end{align}
noting that  $p(\tau_{j,k}^2\,|\,\thetavec\setminus\tau_{j,k}^2,\yvec)$ are univariate inverse gamma distributions with shape and scale given by $a_{j,k}=\tfrac{1}{2}\rank(K_{j,k})$, $b_{j,k}+\tfrac{1}{2}\betavec_{j,k}^\top  K_{j,k}\betavec_{j,k}$. 
This second choice is attractive for three reasons which we formalize below. Throughout, we make the following assumptions.
\begin{assumption}\label{assumption1}
\text{} 
\begin{enumerate}[i.)]
\item The  approximating family for the marginal of $\betavec$ of $q_\lambda^0(\betavec)$ in \eqref{eq:va2} is the same as that of the general fixed form approximation $q_\lambda(\betavec,\tilde\tauvec^2)$.  \label{A11}
\item It is feasible to generate from $p(\tauvec^2\,|\,\betavec,\yvec)$ exactly or approximately. \label{A12}
\end{enumerate}
\end{assumption}

\begin{corollary} The lower bound $\mathcal{L}^0(\lambdavec)$ of the VA at \eqref{eq:va2} of $q_\lambda^0(\betavec)$ is the same as the lower bound at \eqref{eq:ELBO} and thus does not require evaluation of the marginal posterior $p(\betavec|\yvec)$ with $\tauvec^2$ marginalized out, i.e.
$   \mathcal{L}^0(\lambdavec) =\mathcal{L}(\lambdavec).$
\end{corollary}
\begin{proof}
\begin{equation*}\begin{aligned}\label{eq:L}
   \mathcal{L}(\lambdavec) &=\dsE_{q_\lambda}\left\lbrack\log p(\yvec\,|\,\betavec,\tauvec^2)+\log p(\tauvec^2)-\log q_\lambda^0(\betavec)-\log p(\tauvec^2\,|\,\betavec,\yvec)\right\rbrack\\
   &=\dsE_{q_\lambda}\lbrack\log p(\yvec\,|\,\betavec)+\log p(\betavec)+\log p(\tauvec^2\,|\,\betavec)-\log p(\tauvec^2)+\log p(\tauvec^2)-\log q_\lambda^0(\betavec)-\log p(\tauvec^2\,|\,\betavec) \rbrack\\
   &=\dsE_{q_\lambda^0}\left\lbrack\log p(\yvec\,|\,\betavec)+\log p(\betavec)-\log q_\lambda^0(\betavec)\right\rbrack=\mathcal{L}^0(\lambdavec).
\end{aligned}\end{equation*}
\end{proof}
\begin{corollary} Let $\thetavec$ have a VA of the form \eqref{eq:va2}. Let $\zetavec=((\zetavec_\beta)^\top,(\tauvec^2)^\top)^\top$, where $\zetavec_\beta$ is such that $\betavec=t^0(\zetavec_\beta,\lambdavec)$. Assume the density $f_\zeta(\zetavec)=f_{\zeta_\beta}(\zetavec_\beta)p(\tauvec^2\,|\,t^0(\zetavec_\beta,\lambdavec),\yvec)$ is such that it does not depend on $\lambdavec$ and such that there exists a transformation $t$ given by $\thetavec=t(\zetavec,\lambdavec)=(t^0(\zetavec_\beta,\lambdavec)^\top,(\tauvec^2)^\top)^\top$, where $\betavec=t^0(\zetavec^0,\lambdavec)$. Then, the re-parameterization gradient at \eqref{eq:gradELBO2} simplifies to
\begin{equation}\begin{aligned}\label{eq:graddL0}
 \nabla_\lambda\mathcal{L}(\lambdavec)=\dsE_{f_\zeta}\left\lbrack \frac{\partial \betavec^\top}{\partial\lambdavec}\nabla_\beta\left\lbrace \log g(t(\zetavec,\lambdavec))-\log q_\lambda^0(\betavec)-\log p(\tauvec^2\,|\,\betavec,\yvec) \right\rbrace\right\rbrack.
\end{aligned}\end{equation}
\end{corollary}
\begin{proof}
Follows directly from \eqref{eq:gradELBO2}. 
\end{proof}
\begin{remark}
An  important difference to \citet{LoaSmiNotDan2022} is the hierarchy between the two components of $\thetavec$. While in \citet{LoaSmiNotDan2022} the latent variables have a prior conditional on the ``global'' parameters, in our case, the prior of the regression coefficients $\betavec$ is conditional on $\tauvec^2$. Hence, we need the conditional posterior of $\tauvec^2$ and thus consider $\mathcal{L}$ rather than $\mathcal{L}^0$ for optimization. 
\end{remark}

\begin{corollary} Consider a VA of the form \eqref{eq:va2} and a second VA $\tilde q_{\tilde\lambda}(\betavec,\tilde\tauvec^2)$ with the same marginal approximation for $\betavec$. Write $\tilde q_{\tilde\lambda}(\betavec,\tilde\tauvec^2)=q_{\tilde\lambda_{1}}^0(\betavec)q_{\tilde\lambda_2}(\tilde\tauvec^2\,|\,\betavec)$ for this second VA with variational parameters $\tilde\lambdavec=(\tilde\lambdavec_1,\tilde\lambdavec_2)$ and let $\tilde\lambdavec^\ast=(\tilde\lambdavec_1^\ast,\tilde\lambdavec_2^\ast)$ be the optimal variational parameters. Then, our VA at \eqref{eq:va2} with optimal parameter vector $\lambdavec^\ast$ is more accurate than $\tilde q_{\tilde\lambda}(\betavec,\tilde\tauvec^2)$ in the sense of having a lower KL divergence:
\[
\mbox{KL}(q_{\lambda^\ast}^0(\betavec)p(\tauvec^2\,|\,\betavec)\,||\,p(\betavec,\tauvec^2\,|\,\yvec))\leq \mbox{KL}(q_{\tilde\lambda_1^\ast}^0(\betavec)q_{\tilde\lambda_2^\ast}(\tilde\tauvec^2)\,||\,p(\betavec,\tilde\tauvec^2\,|\,\yvec)).
\]
\end{corollary}
\begin{proof} For any VA of the form $q(\betavec,\tilde\tauvec^2)=q(\betavec)q(\tilde\tauvec^2\,|\,\betavec)$ of $p(\betavec,\tilde\tauvec^2\,|\,\yvec)$
\begin{equation}\begin{aligned}
 \small \mbox{KL}(q(\betavec,\tilde\tauvec^2)\,||\,p(\betavec,\tilde\tauvec^2\,|\,\yvec)) = \mbox{KL}(q(\betavec)\,||\,p(\betavec\,|\,\yvec)) + \int \mbox{KL}(q(\tilde\tauvec^2\,|\,\betavec)\,||\,p(\tilde\tauvec^2\,|\,\betavec,\yvec))q(\betavec)d\betavec. \label{eq:Leffic}
\end{aligned}\end{equation}
If  \eqref{eq:va2} is used, the second term on the right-hand side of  \eqref{eq:Leffic} is  zero as $q_{\tilde\lambda_2^\ast}(\tilde\tauvec^2|\betavec)=p(\tilde\tauvec^2\,|\,\betavec)$.  Furthermore, for any other VA of the form at $\tilde q_{\tilde\lambda}(\betavec,\tilde\tauvec^2)=q_{\tilde\lambda_{1}}^0(\betavec)q_{\tilde\lambda_2}(\tilde\tauvec^2\,|\,\betavec)$, by Assumption \ref{assumption1}i.) the approximation to the marginal posterior
distribution of $\betavec$ cannot improve on the KL-optimal approximation within the chosen family $q^0_{\tilde{\lambda}_1}(\betavec)$
 for approximation \eqref{eq:va2}. 
\end{proof}

\subsubsection{Gaussian variational approximations}\label{subsubsec:GVA}
Gaussian distributions as VA are  popular, often providing very accurate approximations to at least posterior means/modes. As to the best of our knowledge, we are the first to suggest VI for Bayesian SADR and our approach is also based on Gaussian VAs for either $q_\lambda^0(\betavec)$ or $q_{\lambda}(\betavec,\tilde\tauvec^2)$. However, despite their tractability, Gaussian VAs can be computationally burdensome or difficult to estimate when the dimension of $\thetavec$ is high or when an unrestricted covariance matrix is employed because the number of elements in the covariance matrix
increases quadratically with the parameter dimension. To overcome this burden in our models where $p_\theta$ is typically large, we follow \citet{OngNotSmi2018} and employ a factor covariance
structure to reduce the number of variational parameters. Specifically, depending on the choices for the VA we made in the previous subsection, we assume that $q_\lambda(\thetavec)$ or $q_\lambda^0(\betavec)$ are of the form $\phi_{p_{\bullet}}(\bullet;\muvec,B^\top B+D^2)$, $\bullet\in\lbrace \thetavec,\betavec\rbrace$, where $B$ is an $p_{\bullet}\times k$ matrix, with far fewer columns than rows, $k\ll p_{\bullet}$, and zeros above the diagonal. Furthermore, $D=\mbox{diag}(d_1,\ldots,d_{p_{\bullet}})$ is a diagonal matrix and $p_{p_{\bullet}}(\cdot;\avec,A)$ is the density of a $p_{\bullet}-$dimensional Gaussian distribution with mean $\avec$ and covariance $\Lambda$. Hence, $\lambdavec=(\muvec^\top,\mbox{vech}(B)^\top,\dvec^\top)^\top$, where $\dvec=(d_1,\ldots,d_{p_{\bullet}})^\top$ and $\mbox{vech}$ is the vectorization of the lower triangular elements of $B$ excluding the diagonal, i.e.~the non-zero elements. This distribution allows the usage of the re-parameterization trick as we can write $t(\zetavec,\lambdavec)=\muvec+B\xivec+(d_1,\ldots,d_{p_{\bullet}})^\top\circ\varepsilonvec$, where $\zetavec=(\xivec^\top,\varepsilonvec^\top)^\top\sim N(\nullvec,I_{(k+p_{\bullet})\times(k+p_{\bullet})})$. \citet{OngNotSmi2018} provide   analytical expressions for the gradients of the log-VA that allow for a fast computation of unbiased gradients. Algorithm~\ref{alg:VI} in Section~\ref{app:algos} of the Supplement summarizes the estimation procedure. 

\subsection{Subsampling}\label{subsec:sub}

Our VI approach can readily be combined with subsampling  which results in doubly stochastic VI \citep{SalKno2013} and which has two important advantages. First, it can help avoid getting stuck in local optima; and second, it can increase efficiency especially for large data sets due to the fact that in each step we do not use every data point  to evaluate the likelihood and gradient estimate but only a subsample of size $n_{\mbox{\scriptsize{sub}}}<n$ \citep{ZhaBueKjwMan2018}.  The implementation evaluates the likelihood only on the subsample and adds a re-weighting factor $\frac{n}{n_{\mbox{\scriptsize{sub}}}}$, which also leads to a weighted derivative \citep[see][for details]{HofBleWanPai2013}. The adjusted algorithm can be found in the Supplement~Algorithm \ref{alg:subsampling}.

\subsection{Global annealing}\label{subsec:glob}
\citet{AbrManRanBle2014} introduce a global annealing procedure called annealed variational inference (AVI). AVI introduces a global temperature to the likelihood to allow for appropriate weighting of the data in computing the posterior distribution.  The basic intuition of the adjusted likelihood is to first down-weight the importance of the data with a reasonable large $T$ and then to  sequentially decrease $T$ during the VI algorithm to force the VA to explain the data \citep{ManMciAbrRanBle2016}. A well-designed annealing schedule can thereby contribute to faster convergence and to avoid local optima in complex models, similar to subsampling.

For our Bayesian SADR, the  extended conditional joint distribution  reads as
\begin{align*}
	p(\betavec,\tilde\tauvec^2,\yvec\mid\xvec,T)&=\frac{p_Y(\yvec\,|\,\xvec,\betavec)^{1/T}p(\betavec\,|\,\tilde\tauvec^2)p(\tilde\tauvec^2)}{C(T\,|\,\xvec) }\propto p_Y(\yvec\,|\,\xvec,\betavec)^{1/T}p(\betavec\,|\,\tilde\tauvec^2)p(\tilde\tauvec^2)
\end{align*}
where $C(T\,|\,\xvec)=\displaystyle\int p_Y(\yvec|\xvec,\betavec)^{1/T}p(\betavec\,|\,\tilde\tauvec^2)p(\tilde\tauvec^2)d\yvec,\d\betavec\, d\tilde\tauvec^2$ is the  normalizing constant. 
The resulting lower bound, called annealed ELBO, can be derived in our case as 
\[	\mathcal{L}(\lambdavec\,|\,T)  = \mathbb{E}_{q_\lambda(\thetavec)}\left[\frac{1}{T}\log( p_Y(\yvec\,|\,\xvec,\betavec))+\log(p(\betavec\,|\,\tilde\tauvec^2))+\log(p(\tilde\tauvec^2))-\log(q_\lambda(\betavec,\tilde\tauvec^2))\right] \] 
and similar in the case of conjugacy in the prior for $\tilde\tauvec^2$.
For $T$=1, the standard ELBO is recovered. The re-parameterization trick and the resulting derivatives for the annealed ELBO are straightforward to compute as only the likelihood needs to be re-weighted by the reversed temperature.  

The annealing schedule has a strong influence on the performance of the annealing procedure.  Following \citet{AbrManRanBle2014}, we use a linear decrease schedule for $T$, and update it every 100 iterations until it reaches $T$=1 at a pre-chosen iteration. 
 For the starting  value $T^{(0)}$ we compare $T^{(0)}\in\lbrace2,5,20,30\rbrace$, similar to \citet{AbrManRanBle2014}. Empirical evidence for the benefit of using global annealing and the corresponding algorithm can be found in the Supplement, Section~\ref{app:defaults} and Algorithm \ref{alg:anneal}, respectively.

\subsection{Further computational details}
VI can be  sensitive to initialization \citep{AltRanBle2018} and carefully chosen starting values can help to prevent getting stuck in local modes of the ELBO or lead to faster convergence.  For $\betavec$ we adapt the MCMC initialization of \citet{UmlKleZei2018} by using the penalized maximum likelihood estimators with fixed $\boldsymbol{\tau}^2$ which is equivalent to the estimation of the posterior mode of $\betavec$, while additionally minimizing information criteria such as the AIC/BIC are used for the smoothing parameters. While MCMC sampling is highly sensitive to initialization, our simulations show that our newly introduced VI algorithm performs similarly if we initialize $\boldsymbol{\beta}^{(0)}=0$.

In addition to ADADELTA, we use an adapted stopping criterion  of the stochastic optimization based on the criterion introduced by \citet{YaoVehSim2018}. Here, the algorithm ends when the change of the estimated moving average of the lower bound did not improve by more than $10^{-4}$ in the last 1000 iterations. Due to critics concerning the instability of this criterion we replace the average with the median. To arrive at a point estimate $\hat\lambdavec$ for the variational parameters $\lambdavec$, we use the last 1000 iterations  to calculate a final point estimate $\hat{\boldsymbol{\lambda}}=\frac{1}{1000}\sum_{\text{iter}=\text{niter}-1000}^{\text{niter}} \boldsymbol{\lambda}_{\text{iter}}$. Uncertainty about model parameters and functions of these can be directly by sampling from full variational density  $\mathcal{N}(\hat{\boldsymbol{\mu}},\hat{\Sigma})$.

\setlength{\abovedisplayskip}{0.0cm}
\setlength{\belowdisplayskip}{0.0cm}
\section{Robust Fitting}\label{sec:extend}
In Sections \ref{subsec:sub} and \ref{subsec:glob}, we discussed adaptations to the likelihood function that introduce a degree of robustness against local optima via subsampling and the introduction of a global temperature.  In this section, we take a different route using the concept of local annealing, initially introduced by \citet{ManMciAbrRanBle2016}. This technique involves assigning distinct temperatures to each data point and treating them as individual random variables within the VI framework. \citet{WanKucBle2017} adapt the idea of local annealing   and present a data driven approach of re-weighting the data by assigning each data point a weight in (0,1) which is treated as an additional model parameter. 

The process of re-weighting data points can be seen equivalent to a data driven detection of data outliers. The influence of extreme values is reduced so that they do not influence the predictions too strongly. This idea is also refereed to as Bayesian data re-weighting. Detecting outliers in complex models is often not straightforward. That is why building a data driven probabilistic model which objectively detects outliers is very helpful in automatically reducing their impact.
 The goal is to down-weight extreme values to make estimations more robust. It is important to note that the parameter space of the optimization problem will be increased by the number of observations in the data set which makes it difficult with exact Bayesian estimation but tractable with VI.


\subsection{Model specification} Following \citet{WanKucBle2017}, the likelihood $p_Y(\yvec\,|\,\xvec,\betavec)$ is augmented by a weight vector $\wvec=(w_1,\ldots,w_n)^\top\in(0,1)^n$. In the Bayesian framework, we treat $\wvec$ as random and assign independent beta priors to each $w_i\in (0,1)$ as priors, $
	p(\boldsymbol{w}) =\prod_{i=1}^n p(w_i)\propto \prod_{i=1}^n w_i^{a_w-1}(1-w_i)^{b_w-1}
$. Depending on the values for the shared parameters $a_w$ and $b_w$, both rather uninformative but also informative priors can be recovered. The latter is  important here as most of the observations should be included in the model with weights close to one, while only   a few weights (those belonging to outlier observations) should be small (thus down-weighting the corresponding $y_i$). To see which prior distribution performs best we conducted a small simulation study presented in the Supplement~Section~\ref{app:simrobust} in the Supplement.

\subsection{Augmented posterior}  The aim is  to target the augmented posterior
 \[p(\betavec,\wvec,\tilde\tauvec^2\,|\,\yvec)\propto \prod_{i=1}^n p_Y(y_i\,|\,\xvec_i,\betavec)^{w_i}\, p(\betavec\,|\,\tilde\tauvec^2)p(\tilde\tauvec^2)p(\wvec).\]
 The reweighted density integrates to one when the normalizing factor $\int p(\betavec\,|\,\tilde\tauvec^2)p(\tilde\tauvec^2)p(\wvec)$ $\prod_{i=1}^n p_Y(y_i\,|\,\xvec_i,\betavec)^{w_i} d\betavec d\tilde\tauvec^2 d\wvec$ is finite. 
For likelihoods of the exponential family, this reweighted density always integrates to one \citep{BerSmi2009}. To ensure a proper posteriors also for SADR models, we assume that conditions (A.1), (A.2), (B.1)--(B.3) and (C.1)--(C.4) of \citet{KleKneLan2015} hold. These conditions are sufficient to obtain a proper posterior \citep[see Theorem 1 of][for details]{KleKneLan2015}.

For posterior estimation using VI, we  transform the weights $\wvec$ to the real line via a logit transformation, i.e.~by setting $\widetilde{w}_i = \text{logit}(w_i)$.

\subsection{VI to the augmented posterior:} 
 Consider the augmented model parameter vector $\boldsymbol{\phi} =  (\thetavec, \tilde{\boldsymbol{w}})^\top$, where $(\widetilde w_1,\ldots,\widetilde w_n)^\top$. 
  Since $\dim(\wvec)=n$, computational cost can be high in particular when $n$ is large. To overcome this computational challenge and to make our VI approach efficient, we assume that the VA is the product of a Gaussian density with factor covariance structure for the parameters of $\thetavec$ times a Gaussian density with a diagonal covariance  for the weights:
  \begin{align*}
	q(\thetavec,\widetilde{\boldsymbol{w}}\,|\,\overbrace{\boldsymbol{\lambda_\theta},\lambdavec_{\widetilde{w}}}^{=:\lambdavec}) & =q_{\lambda_\theta}(\boldsymbol{\theta})\times  q_{\lambdavec_{\widetilde{w}}}(\widetilde{\boldsymbol{w}})=\mathcal{N}(\thetavec;\muvec_\theta,\overbrace{B_\theta B_\theta^\top+D_\theta^2}^{\Sigma_\theta})\times \mathcal{N}(\widetilde{\boldsymbol{w}};\muvec_{\widetilde{w}},\overbrace{\exp(\boldsymbol{\varrho})^2I_N}^{\Sigma_{\widetilde{w}}})\\
	& = \mathcal{N}\left(\begin{bmatrix} \muvec_\theta \\ \muvec_{\widetilde{w}} \end{bmatrix} ,\begin{bmatrix}\Sigma_\theta & 0 \\ 0 & \Sigma_{\widetilde{w}} \end{bmatrix}\right).
\end{align*}
This choice implies the assumption of posterior independence between $\widetilde\wvec$ and $\thetavec$. While this seems a rather strong assumption at first sight, we found it to work well in practice. 

The annealed ELBO to the augmented posterior is given by
\begin{align}
	\mathcal{L}(\lambdavec_\theta,\lambdavec_{\widetilde{w}}) & = \mathbb{E}_{q_{\lambdavec}(\betavec,\tilde{\tauvec}^2,\widetilde{\wvec})}[\underbrace{\log(p(\yvec,\boldsymbol{\beta},\tilde{\tauvec}^2,\widetilde{\wvec})))}_{=:h(\boldsymbol{\phi})}-\log(q_{\lambda_\theta}(\boldsymbol{\theta})) -\log(q_{\lambda_{\widetilde{w}}}(\widetilde{\wvec}))].\label{eq:lowerbound2}
\end{align}
To employ the re-parameterization trick to this ELBO, we need the following proposition.
 \begin{proposition}
 Let $\muvec^{\text{total}}=\left(\muvec_\theta^\top,\muvec_{\widetilde{w}}^\top\right)^\top\in\dsR^{p_\bullet+ n}$, $B^{\text{total}}=\left(B^\top, 0_{k\times n}^\top \right)^\top\in\dsR^{(p_\bullet+n)\times k}$, where $0_{k\times n}$ is a ${k\times n}$  matrix of zeros \footnote{We use sparse matrices which ensure efficiency.}, $\dvec^{\text{total}} = \left(\dvec^\top,\exp(\boldsymbol{\varrho}^\top) \right)^\top\in\dsR^{p_\bullet+ n}$, $\boldsymbol{\epsilon}^{\text{total}}=\left(\boldsymbol{\epsilon}_\theta^\top,\boldsymbol{\epsilon}_{\widetilde{w}}^\top\right)^\top\in\dsR^{p_\bullet+ n}$ and $\xivec^{\text{total}}\in\dsR^{k}$. The model parameters $\phivec$ are then  re-parameterized as 
 \[\boldsymbol{\phi} = \muvec^{\text{total}} +  B^{\text{total}}\boldsymbol{\xi}^{\text{total}} + \dvec^{\text{total}}\circ \boldsymbol{\epsilon}^{\text{total}} \qquad \text{with } \zetavec=(\boldsymbol{\xivec}^{\text{total}},\boldsymbol{\epsilon}^{\text{total}})^\top\in\dsR^{k+p_\bullet+ n}\sim N(\nullvec,I),\]
 such that $\zetavec$ has a distribution being independent of the variational parameters $\lambdavec$ as desired.
 \end{proposition}

For initialization of the additional variational parameters, we set $\mu_{\widetilde{w}_i}=0.98$  and $\varrho = 1$ following \citet{KucTraRanGelBle2017}. 
Algorithm \ref{alg:robust} in the Supplement~summarizes the resulting VI algorithm with robust fitting.

\setlength{\abovedisplayskip}{0.0cm}
\setlength{\belowdisplayskip}{0.0cm}
\section{Simulations}\label{sec:simul}
In this section we evaluate the performance of our method denoted as \texttt{abamlss} empirically. The aims are to (i) facilitate tuning of \texttt{abamlss}; (ii) compare \texttt{abamlss} to exact Bayesian inference in smaller data sets, where MCMC is feasible (estimation realized via the \texttt{bamlss} package); and (iii) benchmark the robust estimation procedure of \texttt{abamlss} from Section~\ref{sec:extend} against those implemented in \texttt{bamlss} and the penalized likelihood based approach of \citet{AebCanMarRad2021} (denoted as \texttt{robust gamlss}). 
We restrict ourselves to the main simulation design and overall results, whereas further simulation experiments and full details can be found in the Supplement,~Section~\ref{app:simulation}.
\subsection{Simulation design}
Due to the complexity of distributional regression models in general, it is usually not straightforward to setup a reasonable simulation design. To obtain realistic scenarios, we follow \citet{SmiStaKle2021} and base our simulation on  four real data sets of different sample sizes, complexity and response types, see Table~\ref{tab:data sets}. 
\begin{table}\caption*{Simulation study.}
	\begin{tabular}{p{1.5cm}p{1cm}p{5cm}p{2.7cm}p{4.8cm}}
		\hline 
  \hline
  data set & $n$ & Covariates & Response & Source \\
  \hline
\emph{Rents}    & 3,082 & area (m\textsuperscript{2}), year of construction, central heating, quality of  bath/kitchen, districts of Munich & net rent (EUR) & \citet{StaRigBas2021} \\
\emph{Zambia} & 4,847 & age/gender, mother's bmi/ employment/education& stunting & \citet{UmlAdlKneLanZei2015}\\
\emph{Fatalities} & 1,087 & weeks & fatalities & \citet{UmlKleSimZei2019}\\
\emph{Brain} & 1,567 &voxel coordinates &  median FPQ & \citet{Woo2006}\\
  \hline
   \hline
    \end{tabular}
    \caption{Details and source of the four data sets.}
    \label{tab:data sets}
\end{table}
Histograms and distributions of the responses can be found in Figure~\ref{fig:hist_y} in the Supplement~\ref{app:simulation}. The data set \emph{Brain} is used to evaluate the performance of the robust estimation algorithm; see Section~\ref{sec:appl} for details on this data set. We  fit  \texttt{abamlss} to each data set and define the resulting fitted model as data generating processes (DGPs) to simulate replications for evaluation of our and the benchmark methods.   From each of the four DGPs, we simulate 40  data sets of the same sample size as in the original data. A further 41-st data set is generated to evaluate out-of-sample performance. As response distributions, we choose  the gamma and Gaussian distributions for \emph{Rents}. For \emph{Zambia} we choose the Gaussian distribution, for \emph{Fatalities} the Box Cox power exponential distribution and for \emph{Brain} the gamma distribution.

To achieve aim (iii), we follow the simulation approach of \citet{AebCanMarRad2021} and create  contaminated versions of the \emph{Brain} data set replicates by adding artificial outliers. Specifically, we increase the observed response by 10 for 5\% of the data points randomly chosen from the area $Z_1>70$ and $Z_2>30$, where $(Z_1,Z_2)$ are the coordinates of the voxels. 
For the data sets \emph{Rents}, \emph{Zambia}, \emph{Fatalities},  we use $k=5$  for the factor covariance structure and $k=35$ for the \emph{Brain} data set due to the considerably increased parameter space. These values have been chosen based on comparing the mean lower bounds for different $k \in \lbrace 1,2,5,10,15,20,30,40,50\rbrace$. 
\paragraph*{Measures of performance}
Performance is evaluated on the 41-st data set using the log score (LS) and the continuous ranked probability score   \citep[CRPS;][]{GneTilRaf2007} to quantify the accuracy of the entire forecast distributions. Each measure is oriented such that smaller values indicate a better performance. To evaluate the robust estimation methods, we compare the LS and CRPS of robust and non-robust methods not only based on original data set replicates but also on the contaminated data set replicates.

\subsection{Main results}
\paragraph{Tuning \texttt{abamlss}}
We conducted a large number of experiments where we compared different settings  for $M$, the priors  for $\boldsymbol{\tau}_{jk}^2$ and $n_{\mbox{\scriptsize{sub}}}$. The full list of settings and detailed results can be found in Section~\ref{app:defaults} of the Supplement. 
The most important conclusions are as follows. First, as expected, increasing $M$ improves the estimation performance compared to $M=1$ but decreases the time efficiency.  The inclusion of a Gibbs step solves this issue as the algorithm converges faster. However, for smaller parameter spaces $M=1$ is sufficient. Subsampling performs similarly to the standard method and the Gibbs step method, but it outperforms them in time for the Zambia nutrition data set as can be seen in Table~\ref{tab:times}. The difference between the inverse gamma and SD priors for $\boldsymbol{\tau}_{jk}^2$ is small. 
\begin{table}[htbp]
\small{
\centering	\caption*{Simulation study.}\begin{tabular}{l|lll}
\hline\hline
		Method & \emph{Rents} & \emph{Zambia} & \emph{Fatalities} \\
		\hline 
		\texttt{bamlss} & 22.34 &1.56& 1.49\\
		\texttt{abamlss} &7.50  &1.43& 2.72\\
		\texttt{abamlss with SD prior} &8.68 &1.32&2.68\\
				\texttt{abamlss with subsampling}&6.24 &\fbox{\textbf{1.19}}&2.12\\
		\texttt{abamlss with Gibbs} &\fbox{\textbf{2.98}} &1.38& \fbox{\textbf{0.99}}\\
		 \texttt{abamlss with Gibbs \& M=5}& 8.40 &3.12&4.79\\
		 		\texttt{abamlss with M=5}& 19.22 &4.68& 7.50\\
		\hline\hline
	\end{tabular}\caption[Computing times (min). Exact vs.~approximate]{\small{Reported are the average computation times in minutes of 10 replicates for each estimation setting for each data set.}}\label{tab:times}}
\end{table}
\begin{figure}[H]
\centering
\caption*{Comparison with MCMC.}
\includegraphics[scale = 0.35]{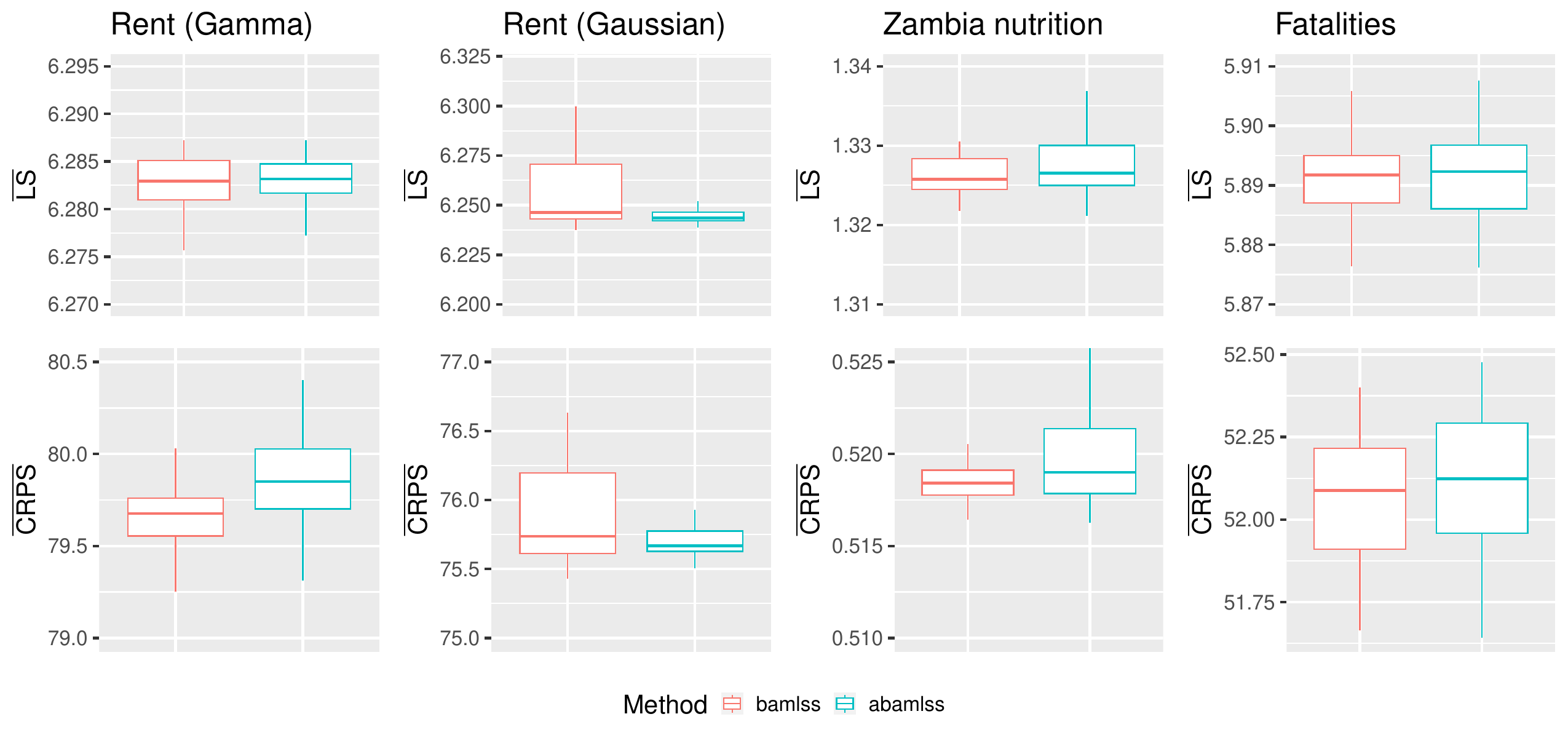}
\caption[Simulation results]{\small{For each data set, the mean LS and  CRPS are depicted for 40 replications for  \texttt{bamlss} (red) and \texttt{abamlss} (blue) with $M=5$ and Gibbs updates.}}\label{fig:bamlss_abamlss}
\end{figure}

\paragraph{Comparison with  \texttt{bamlss}}
 The simulation results presented in Figure~\ref{fig:bamlss_abamlss} show, that even though \texttt{abamlss} is an approximation it always has almost the same accuracy as its exact alternative, while  decreasing the computation time as can be seen in Table~\ref{tab:times}. For the \emph{Rents} data set with Gaussian response it even outperforms \texttt{bamlss}. The largest computational improvement can  be observed for the \emph{Rents} data set where, the default setting only needs three minutes while \texttt{bamlss} takes around 22 minutes.
 
\paragraph{Evaluation of \texttt{abamlss} with robust fitting}
\begin{figure}[H]
\centering \caption*{Robust fitting.}
 			\includegraphics[scale=0.35]{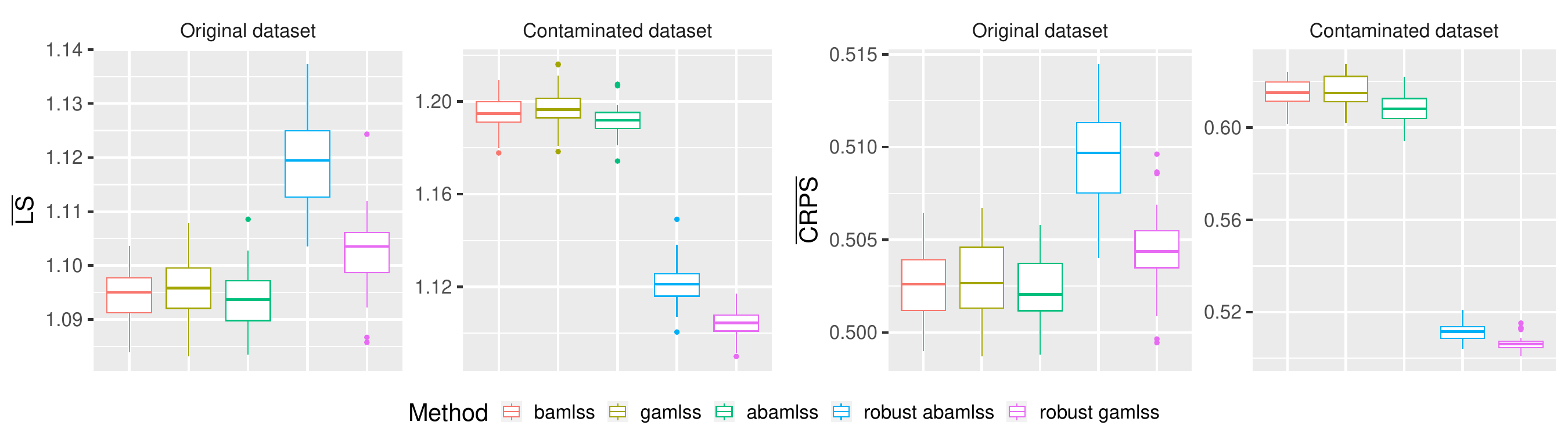}
\caption{\small{The mean LS and  CRPS are depicted for 40 replications for the unchanged \emph{Brain} data set and a contaminated version of it. We compare the Bayesian data reweighting (\texttt{robust abamlss}) with \texttt{abamlss}, \texttt{bamlss} and \texttt{gamlss} and \texttt{robust gamlss} for both data sets.}}\label{fig:robust}
 \end{figure}

As described above we use the unchanged \emph{Brain} data set and additionally a contaminated data set to evaluate the estimation performance of the robust fitting approach of \texttt{abamlss}. As hyperparameter for the prior distribution of the weights, we found  $a_w=0.2$ and $b_w=0.01$ to be optimal as can be seen in Figure~\ref{fig:brain_hyper} in the Supplement, Section~\ref{app:simrobust}. 

We compare \texttt{robust abamlss} with \texttt{abamlss}, \texttt{bamlss}, \texttt{gamlss}, and \texttt{robust gamlss}. From Figure~\ref{fig:robust} we make the following observations. \texttt{Robust abamlss} and \texttt{robust gamlss} outperform all other methods for the contaminated data set. \texttt{Robust gamlss} slightly outperforms \texttt{robust abamlss}, while being slightly slower as can be seen in Table~\ref{tab:times2} in the Supplement. \texttt{abamlss} performs as good as the exact method \texttt{bamlss} and \texttt{gamlss}. Table~\ref{tab:times2} shows that \texttt{robust abamlss} is computationally more efficient than \texttt{robust gamlss}, especially for the contaminated data set.
\paragraph*{Summary}
\begin{itemize}
    \item Comparing the different versions of the \texttt{abamlss} algorithm we find that for all data sets, the inclusion of the Gibbs step in combination with  $M$=5 MC samples has highest accuracy. For less complex model specifications $M$=1 is sufficient. 
   \item Even though \texttt{abamlss} is an approximate method performance is comparable to the exact benchmark \texttt{bamlss}. Yet, \texttt{abamlss} is computationally more efficient.
    \item The comparison of \texttt{robust abamlss} with other methods shows that it performs similarly well as the state of the art robust method \texttt{robust gamlss}, making it a competitive method. Detailed results can be found in Supplement~\ref{app:simrobust}.
\end{itemize}

\setlength{\abovedisplayskip}{0.0cm}
\setlength{\belowdisplayskip}{0.0cm}
\section{Real Data Illustrations}\label{sec:appl}
We now illustrate our VI approach on two real data examples. The first uses data from COVID-19 to model infectious outbreaks. The second employs our robust fitting approach to detecting outliers in analysing brain activity.
\subsection{Modeling infectious outbreaks}
In an era characterized by advancing technology and growing computational power, the ability to predict and forecast the occurrence of diseases has become an invaluable asset in the field of healthcare. The prediction of any malady, whether it be infectious outbreaks like COVID-19 or other diseases, holds profound importance for several critical reasons. 
\paragraph{Data set description}
As an illustration, we consider data used in \citet{SchNicKau2021} and model newly registered German COVID-19 cases based on a smooth time trend, the inclusion of a country-specific discrete spatial effect including 412 districts in Germany and age-gender group effects with the age groups (00--04, 05--14, 15--34, 35--59, 60--79, 80+ years). Additionally, the model allows for delayed registrations for a specific date up to 7 days and includes a weekday dummy to control for the fact that, e.g., mondays will have a higher case numbers caused by delayed reporting over the weekend. Following \citet{SchNicKau2021} we use a total time frame of 21 days for predictions of the next day. Figure~\ref{fig:covid} in Supplement~\ref{app:emp} shows a  histogram of the discrete response. With 622,944 observations the data set is rather large such that exact Bayesian inference would be very time consuming. 

\paragraph{Model specification}
We consider  the count $N_{t,d,r,g}$ of newly registered infections on day $t$ in district $r$ and age-gender group $g$, reported on day $t+d$ ($d$ for the delay with $d\in \lbrace 1,\ldots,7\rbrace$, as response variable \citep{SchNicKau2021}. Due to the acknowledged overdispersion, we model the distribution of the response using a negative binomial distribution with location parameter $\mu=\exp(\eta_1)$ and dispersion parameter $\delta=\exp(\eta_2)$. Each parameter $\eta_k$, $k=1,2$ is linked to a structured additive predictor of the form
$ \eta_{t,d,r,g} = s_1(\mathit{t}) + s_2(lon,lat) + s_3(\mathit{recent days,r}) + \beta_g x_g +\phi \log(1+C_{t-1,d,r,g}) +\mathit{offset}_{r,g},$ 
where
 $s_1(t)$ is a smooth time trend of $t$ modeled with a Bayesian P-spline,
 $s_2(lon,lat)$ is a  spatial effect using thin plate splines,
  $s_3(\mathit{recent\,days,r})$ is a short/long term random district-specific random intercept depending on whether the dummy variable $\mathit{recent\,days}$ is 1 (for the time period within the last 7 days) or 0 (otherwise),
 $\beta_g x_g$ are age-gender group specific random intercepts,
 $\phi\log(1+C_{t-1,d,r,g})$ captures the time-related autoregressive component of the process with $C_{t,d,r,g} = \sum_{j=1}^dN_{t,j,r,g}$ representing the cumulative count over $d$ per age-gender group and district,
 and $\mathit{offset}_{r,g}$ are district and age-gender specific offsets.

\begin{figure}[htbp]\caption*{COVID outbreaks.}
				\hspace{-0.3cm}	\begin{subfigure}{0.5\textwidth}
			\includegraphics[scale=0.15]{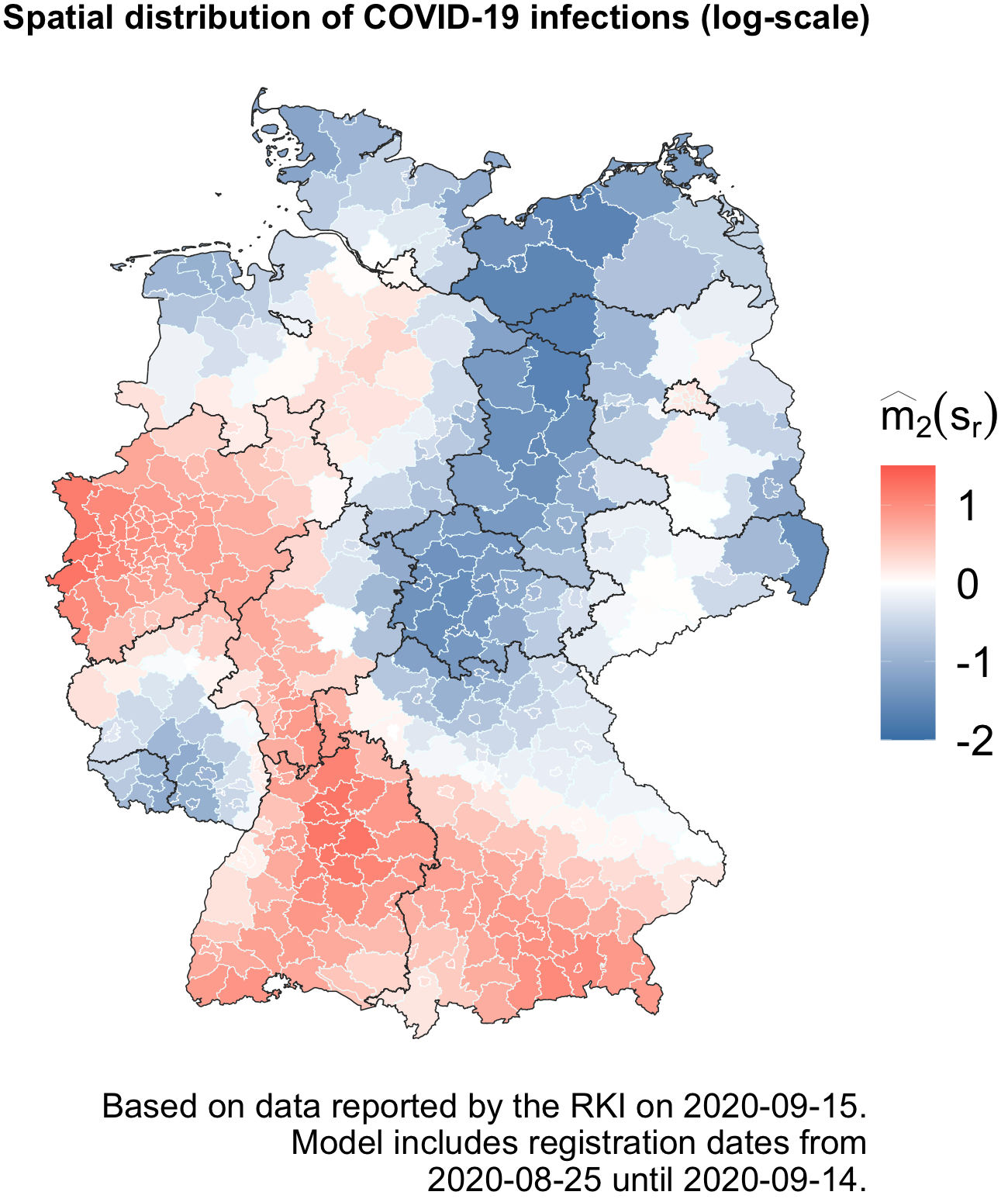}\vspace{0.5cm}\end{subfigure}\begin{subfigure}{0.5\textwidth}\includegraphics[scale=0.40]{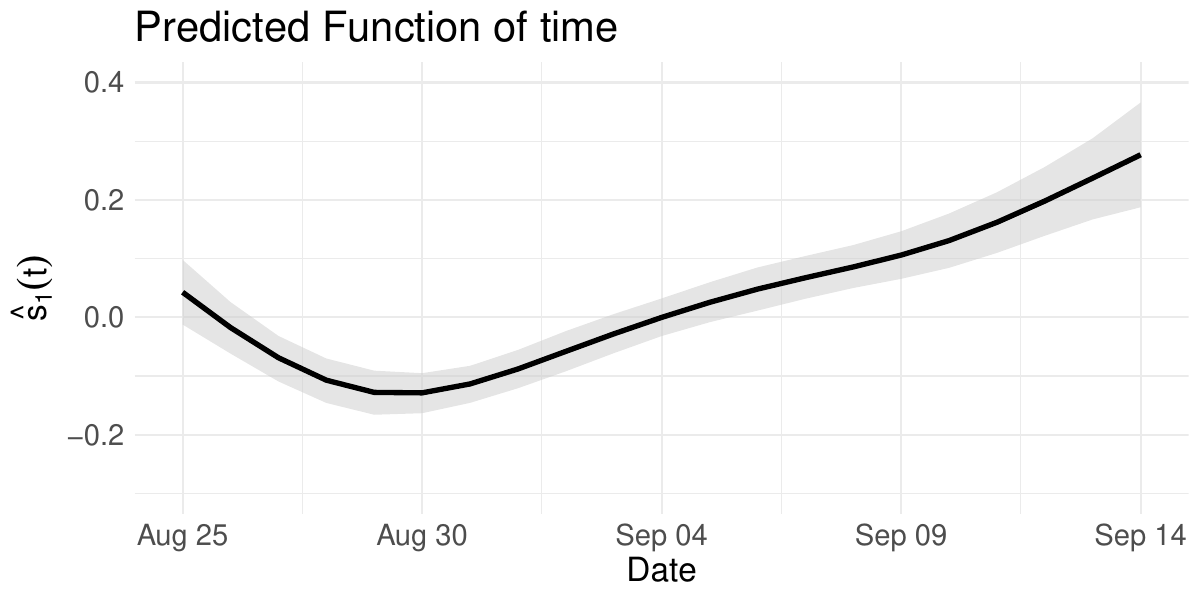}\vspace{0.5cm}\end{subfigure}\hspace{-0.4cm}
   \begin{subfigure}{0.5\textwidth}
			\includegraphics[scale=0.15]{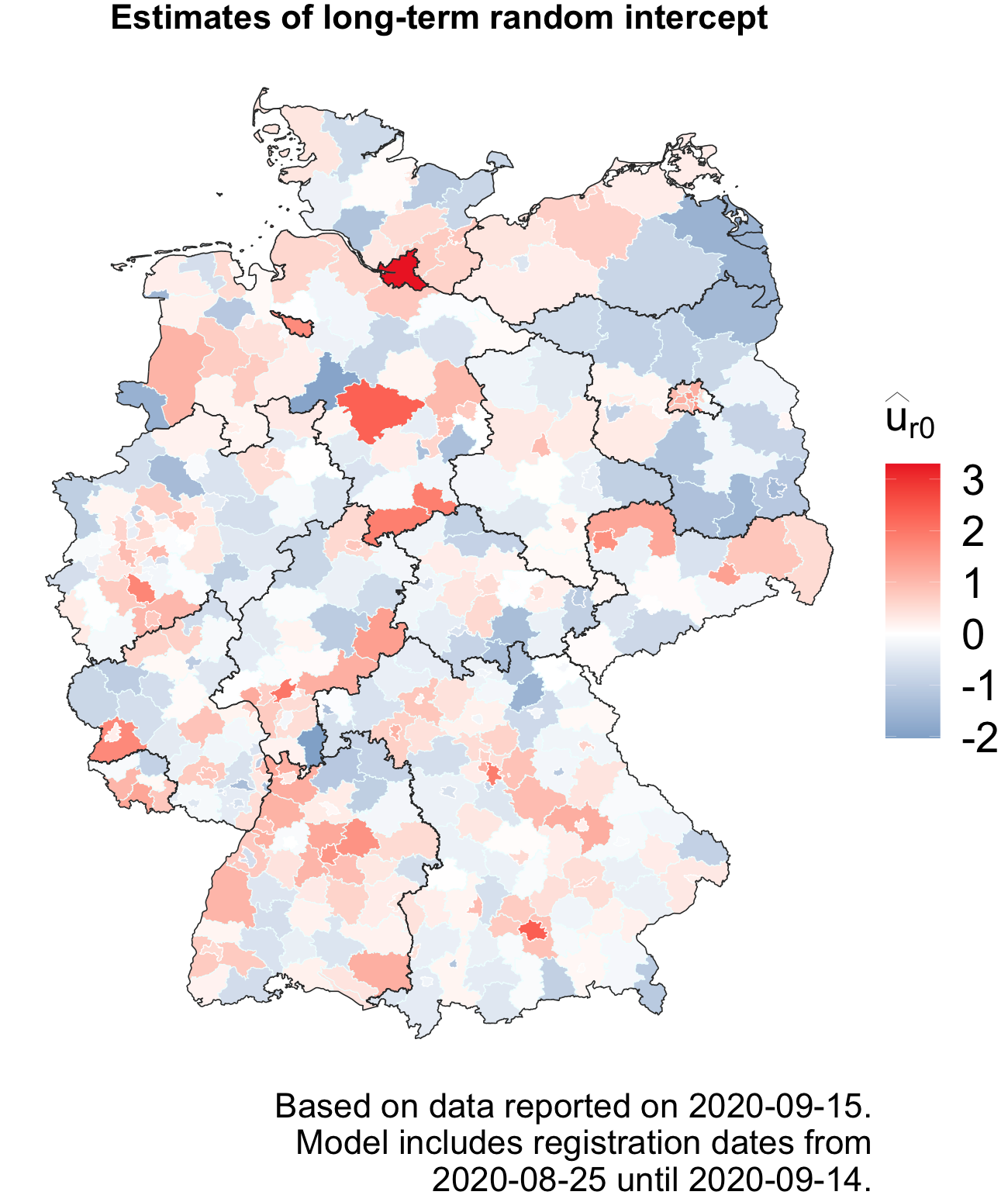}\end{subfigure}\hspace{0.2cm}\begin{subfigure}{0.5\textwidth}\includegraphics[scale=0.15]{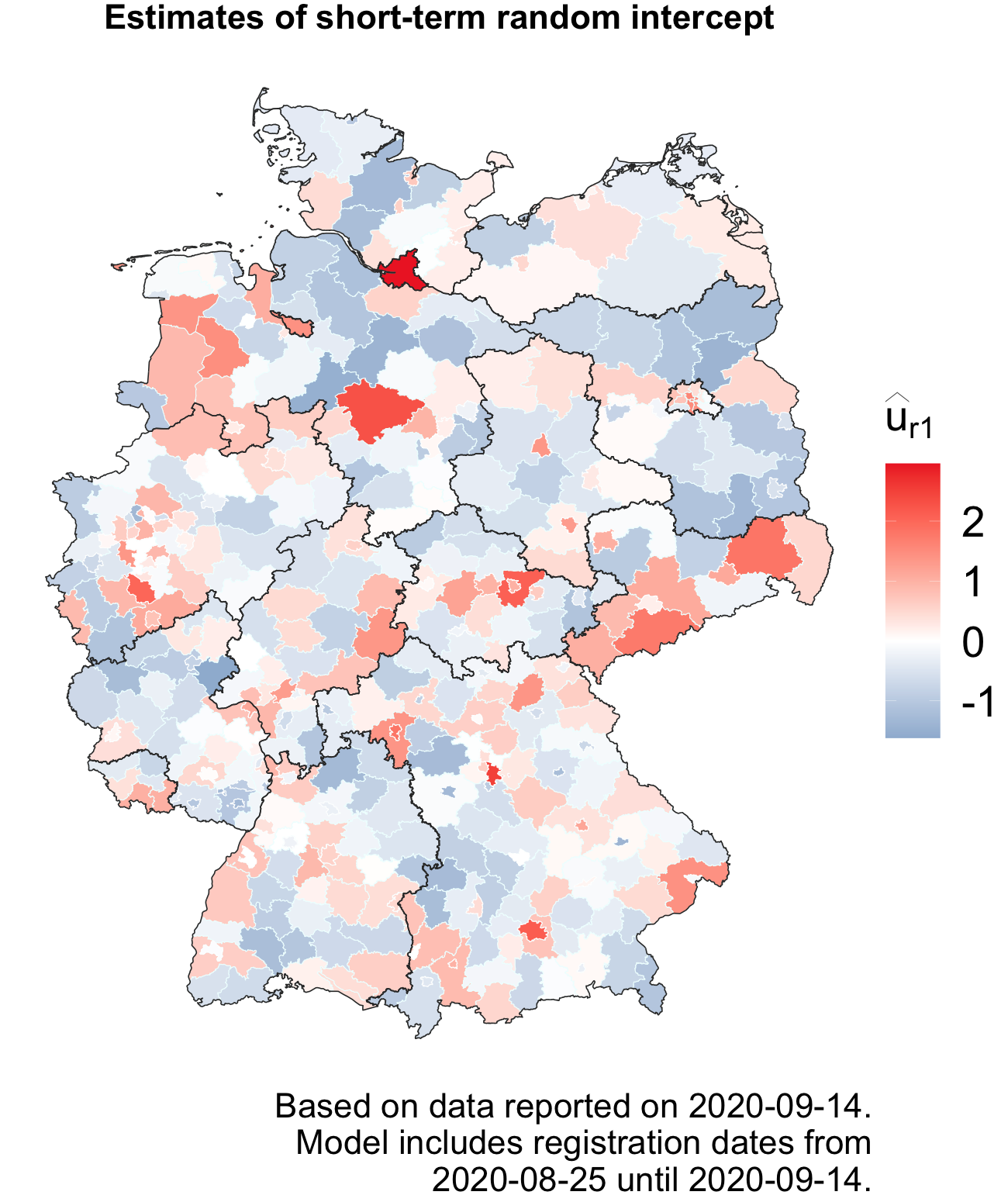}\end{subfigure}
			\caption{Posterior mean estimates of the smooth spatial effects (top left), of the smooth effect of time (top right), the region-specific effects (bottom left) and dynamics (bottom right) of COVID-19 infections based on data from August 25  until September 14 (21 days prior to September 15). For the spatial plots, we reused code generously shared by \citet{SchNicKau2021} in their GitHub repository. \url{https://github.com/gdenicola/Now-and-Forecasting-COVID-19-Infections}}\label{fig:covd1}\end{figure}

\paragraph{Results}
Similar to  \citet{SchNicKau2021}, we  focus our examination on the 15th of September, which preceded the second wave of the COVID-19 pandemic in Germany. 

Figure~\ref{fig:covd1} depicts the estimates of the included smooth effects. In line with the findings reported by \citet{SchNicKau2021}, we observe a decrease in cases until the end of August, which subsequently transitions into a steady and continuous increase leading up to the analysis date (cf.~upper right panel of Figure~\ref{fig:covd1}
The estimated smooth spatial effect for this specific date is depicted in the upper left panel. Notably, regions such as North Rhine-Westphalia, Bavaria, and Baden-W\"urttemberg exhibited strong positive impacts. The upper right panel shows the time trend from the last 21 days. A decline in documented cases was apparent until the conclusion of August, but subsequently, there was a resurgence in the weeks that followed, resulting in an overturn and a consistent uptick. 
The second row of the figure showcases district-specific random intercepts. In the left-hand panel, we assess the relative infection situation in the 21 days leading up to the analysis date. In the right-hand panel, we focus on the district-specific situation in the last 7 days. This analysis aids in identifying districts where the infectious situation has undergone recent changes, enabling authorities to be better prepared to provide support to these areas. For example, on that day, Hamburg has relatively high short- and long term effects which concludes in a relatively high infectious rate  within the last 7 days. In eastern Lower Saxony, we are observing a concerning trend characterized by a swift increase in the number of cases. This escalation has intensified over the past 7 days when compared to the preceding three-week period. We observe a similar trend for most of the districts in Saxony.

\subsection{Robust fitting}

\paragraph{Data set description}
 The \textit{Brain} data set from the R package \texttt{gamair} was previously studied in \citet{LanEllBul2004} and contains coordinates of each of the 1,567 voxels for the median of three measurements of the fundamental power quotient (FPQ) of $n=254$ study patients' brains as covariates. The FPQ is the brain response acquired during a study where a healthy participant had to generate words beginning with a cued letter, while the baseline condition was given by covertly repeating a letter. We consider the median FPQ (medFPQ) as reponse variable resulting in 1,567 oberservations (one per voxel).
 
16 observations have a median FPQ larger than 7, which results in a highly right skewed marginal distribution, see Figure~\ref{fig:hist_y} (lower right) in Supplement~\ref{app:simulation} for a histogram of the response. Nine of these observations are located in the upper right corner ($Z_1>70$ and $Z_2>30$), which typically negatively influence standard prediction algorithms due to overfitting or numerical instabilities. 
Most earlier analyses to this data set excluded two observations which are located in the middle on the back of the brain ($Z_2$ close to 10 and $Z_1$ close to 65). We will however not exclude these outliers to test the  robust version of our VI approach. 

\paragraph{Model specification}
Due to the right skewness and the strict positivity of response , we use a gamma distribution with $\mu$ representing the mean and $\sigma$ the shape. We model the influence of the voxels on the median FPQ by using tensor product P-splines with 10 knots in each direction. 
The hyperparameters of the beta distribution of the weights are set to $a_w = 0.2$ and $b_w = 0.01$, see Figure~\ref{app:simrobust} in  Supplement~\ref{app:simrobust} for a justification.

\paragraph{Results}
Figure~\ref{fig:brainapp} presents the predicted logarithmic median FQP  $\log(\mu)$ (first row) and the logarithmic shape parameter of the predictive distribuion $\log(\sigma)$ (second row) for \texttt{abamlss} (left) and \texttt{robust abamlss}. 
\begin{figure}[t]\caption*{Brain activity.}
				\hspace{-0.3cm}	\begin{subfigure}{0.5\textwidth}\caption{abamlss}
			\includegraphics[scale=0.5]{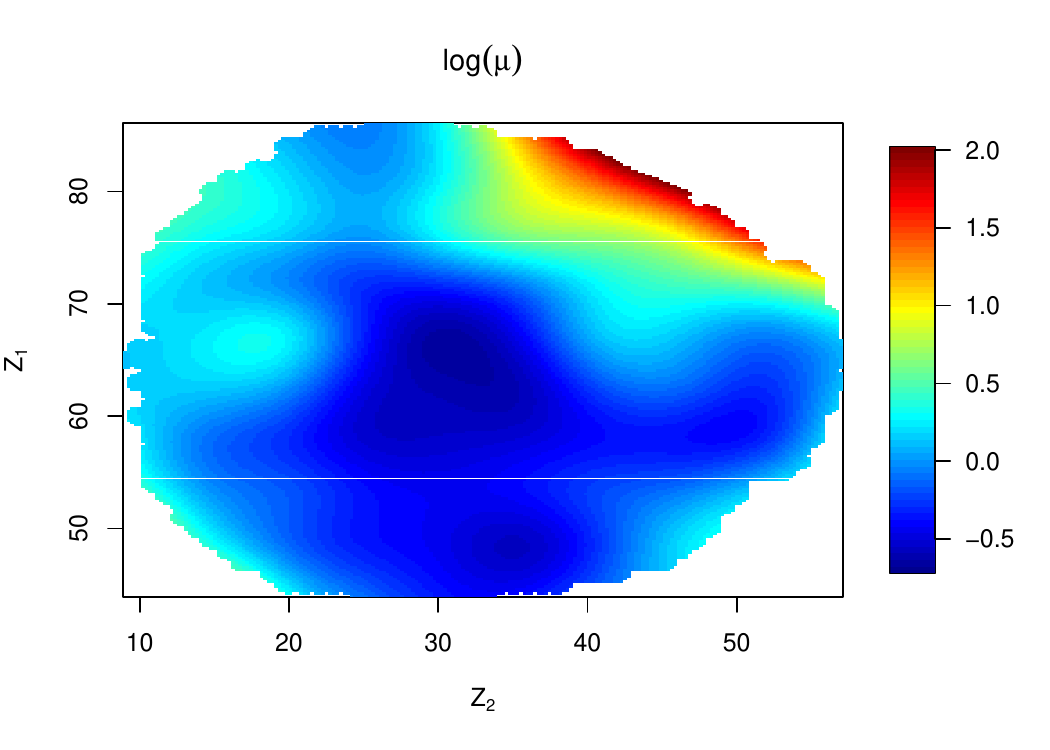}\end{subfigure}\hspace{0.5cm}\begin{subfigure}{0.5\textwidth}\caption{robust abamlss}\label{fig:brainapp_2}\includegraphics[scale=0.5]{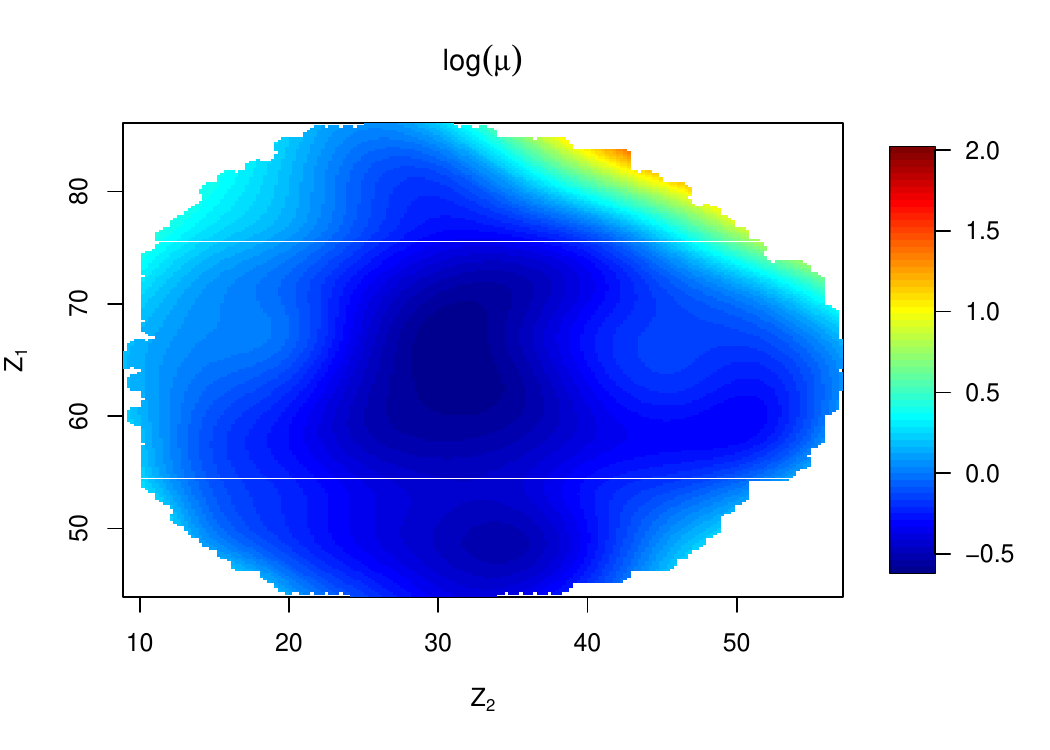}\end{subfigure}
     \begin{subfigure}{0.5\textwidth}\includegraphics[scale=0.5]{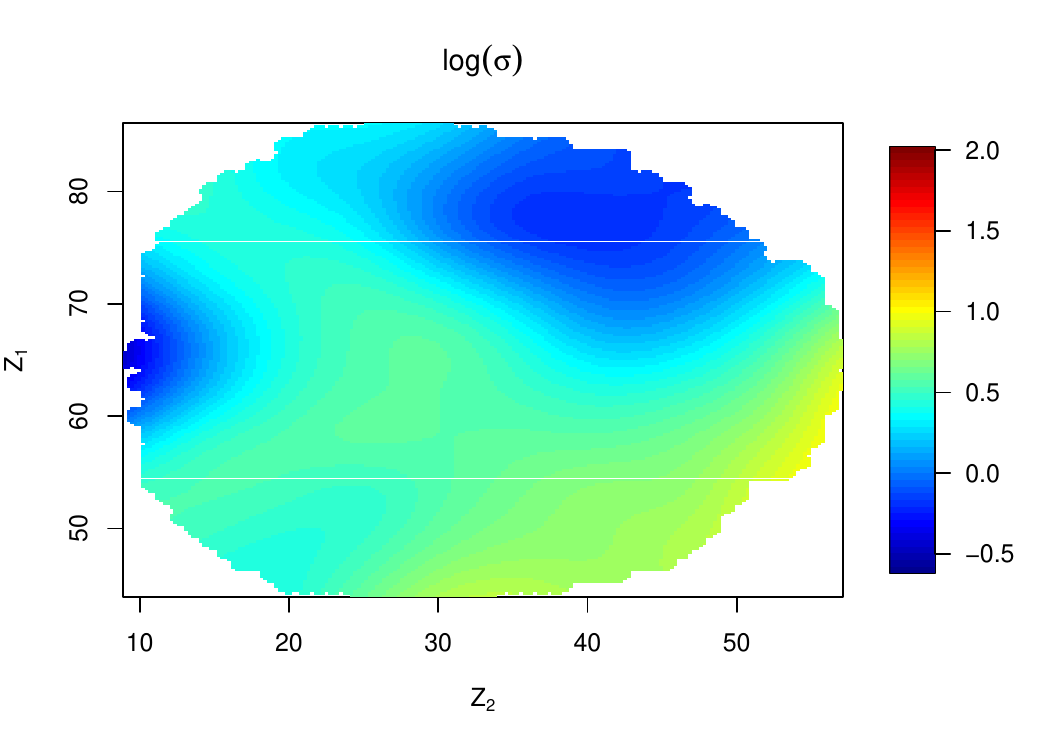}\end{subfigure}
     \hspace{0.5cm}\begin{subfigure}{0.5\textwidth}
					\includegraphics[scale=0.5]{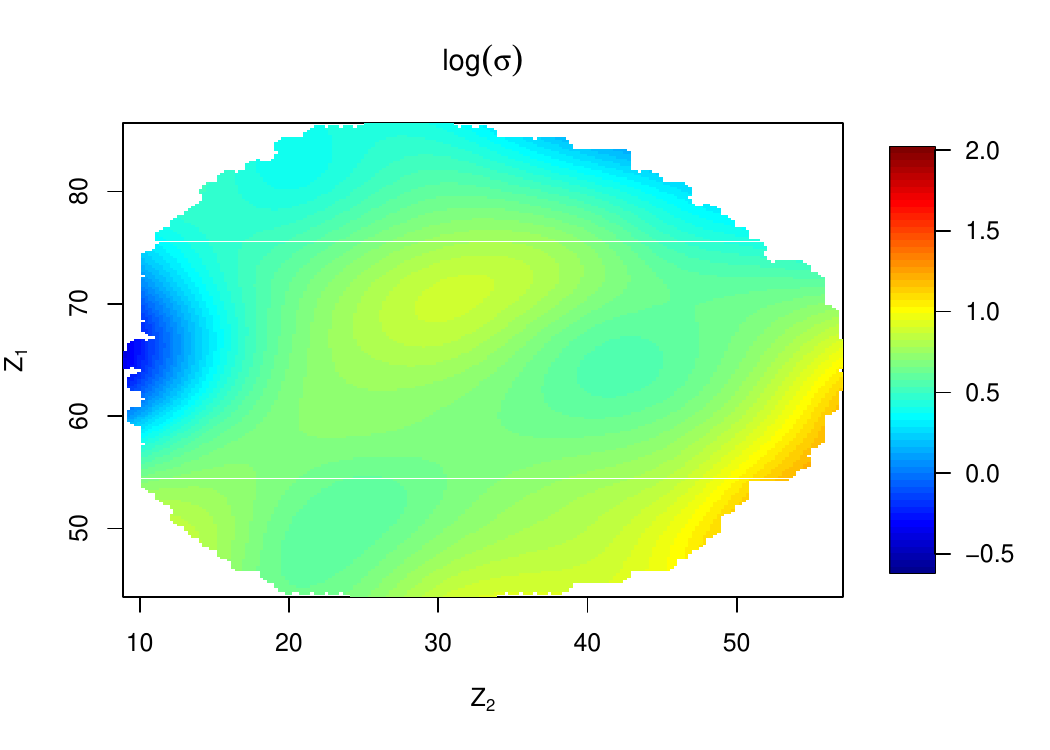}\end{subfigure}
			\caption{Posterior mean estimates of surfaces of the voxel coordinates on  $\log(\mu)$ (top row) and $\log(\sigma)$ (bottom row) using  \texttt{abamlss} (left column) and \texttt{robust abamlss} (right column).}\label{fig:brainapp}\end{figure}
While \texttt{bamlss} predicts high values in the upper right corner and the left center, the robust version predictions are rather smooth and less extreme in that area. The results are similar to \citet{AebCanMarRad2021}. As a result, the robust estimation method succeeds in down-weighting the extreme values as desired. 
\setlength{\abovedisplayskip}{0.0cm}
\setlength{\belowdisplayskip}{0.0cm}
\section{Discussion}\label{sec:discussion}

This paper developed VI for the analysis of large scale distributional regression models within the context of SADR/GAMLSS. Our approach comes with a number of merits and solutions in comparison to existing methods. Firstly, our VI framework does not require prior conjugacy and is not restricted to the class of GLMs/GAMs that model the conditional mean of a response variable only. Secondly, while estimation of fairly flexible SADR models using exact Bayesian inference with MCMC is often not only inefficient  but also time consuming and computationally demadning, the variational approach pushes these limits, allowing us to analyse large-scale distributional models  with many observations and complex predictor structures.  Thirdly, compared to non-Bayesian proposals such as backfitting or SGD type methods, VI comes with direct measures for uncertainty quantification and thus the direct availability of e.g.~credible intervals.  

Overall, we demonstrated the  potential of variational methods in the context of distributional models and extend it to allow for automatic outlier detection through a local annealing procedure. 

In the future, it would be of interest to leverage this potential further to e.g. ~integrate automatic variable selection in SADR \citep{KleCarKneLanWag2021}. One challenge here is the partly discrete model space that renders standard VI assumptions such as a Gaussian VA inappropriate but requires refinements of the variational setup\citep{Guo2023}.
In addition, more work needs to be done to quantify closeness of the VA to the true posterior and to extend our appraoch to more flexible VAs, such as  based on implicit copulas \citep{SmiLoa2022} or mixtures \citep{GunKohNot2021}. Last, it may be interesting to  to see how useful black box VI \citep{KucRajGelBle2015} can be for the general purpose of scalable estimation for SADR models.


\bibliographystyle{Chicago}

\bibliography{references}

\clearpage
\setcounter{page}{1}
\thispagestyle{empty}
\bigskip
\begin{center}
{\sf\LARGE\bf SUPPLEMENTARY MATERIAL\\}
\bigskip
to\\
\bigskip
{\Large\bf Scalable Estimation for Structured Additive
Distributional Regression Through
Variational Inference\\}
\bigskip
by\\
\bigskip
{\large Jana Kleinemeier and Nadja Klein}
\end{center}
\bigskip
\bigskip
\bigskip
\begin{description}
\item[\textbf{\ref{app:algos}:}] Details on algorithm variants of VI for Bayesian SADR, including subsampling, global annealing and robust fitting.
\item[\textbf{\ref{app:simulation}}:] Details on the simulations studies.
\item[\textbf{\ref{app:emp}:}]
Details on the real data illustrations.
\end{description}

\newpage 

\setcounter{page}{1}
\setcounter{figure}{0}  
\setcounter{table}{0}  
\appendix

\renewcommand{\thesection}{\Alph{section}}
\renewcommand\thefigure{\Alph{section}.\arabic{figure}}
\renewcommand\thetable{\Alph{section}.\arabic{table}}

\noindent
\setcounter{page}{0}
\section{Algorithms}\label{app:algos}
In this section, we present the detailed algorithms of the newly introduced estimation approach and its variations. 

Algorithm~\ref{alg:VI} outlines the fundamental estimation procedure, providing a solid foundation for our approach.

When utilizing variational inference to approximate the posterior distribution through stochastic gradient ascent, one often encounters the challenge of getting stuck in suboptimal solutions known as poor local optima. To mitigate the problem of poor local optima, we introduce Algorithm~\ref{alg:subsampling} and Algorithm~\ref{alg:anneal}, both of which modify the likelihood slightly to navigate more effectively within the optimization landscape.

Going even further, Algorithm~\ref{alg:robust} integrates Bayesian data reweighting in the context of SADR. This advanced technique incorporates a random variable for each observation, offering a more robust approach to address the challenges posed by local optima in the context of variational inference.
\begin{algorithm}[H]
\setstretch{1.35}
\caption{\setstretch{1.35}\smallskip VI algorithm for Bayesian SADR.}\label{alg:VI}
\smallskip Set $t=0$. Initialize $\lambdavec^{(t)}=((\muvec^{(t)})^\top, (\text{vech}(B^{(t)}))^\top , (\dvec^{(t)})^\top)^\top$
\begin{algorithmic}[1]
\While{Stopping rule is not satisfied} \nonumber
\State Set iteration $t=t+1$
\State Generate $m=1,...,M$ samples:
$\quad \boldsymbol{\epsilon}_m^{(t)} \sim \mathcal{N}(\nullvec,I_{p_\bullet})$ and $\boldsymbol{\xi}_m^{(t)} \sim \mathcal{N}(\nullvec,I_{k})$
\State Compute $\boldsymbol{\theta}_m^{(t)} = 
 	    \begin{cases}
 	     &  t^0(\boldsymbol{\zeta}_m^{(t)},\boldsymbol{\lambda}^{(t)}) = \boldsymbol{\beta}_m^{(t)} \mbox{ if conjugacy for $\tauvec$ holds}\\
 	     & t(\boldsymbol{\zeta}_m^{(t)},\boldsymbol{\lambda}^{(t)}) \mbox{ otherwise}
 	    \end{cases}$
\If{conjugacy for $\tauvec$ holds}
\State{Generate $m=1,...,M$ samples of $(\tauvec^{(t)}_{jk,m})^2$
\Statex$ \qquad \qquad (\tauvec^{(t)}_{jk,m})^2\sim \text{IG}(a_{jk}+\frac{1}{2}\text{rk}(K_{jk}),b_{jk}+\frac{1}{2}\boldsymbol{\beta}_{jk,m}^{(t) \top}K_{jk}\boldsymbol{\beta}_{jk,m}^{(t)}) \qquad \forall j=1,\ldots,J_k$, $k=1,\ldots,K$}
\EndIf
\State Construct unbiased estimates 
\begin{equation*}
\nabla_\lambda \widehat{ \mathcal{L}(\boldsymbol{\lambda}^{(t)})}=\frac{1}{M}\sum_{m=1}^M \left( \frac{d (\thetavec^{(t)})^\top}{d\boldsymbol{\lambda}^{(t)}}\nabla_{\theta^{(t)}} \left[\log (p(\thetavec^{(t)})\right.\right.
 		{\left.\left.  +\sum_{i=1}^{n}\log \sum p(y_i|\thetavec^{(t)}) - \log q_{\lambda^{(t)}}(\thetavec^{(t)})\right]\right)}
 		\end{equation*}
 		
\State Compute $\rhovec^{(t)}$ using ADADELTA
\State Set $\boldsymbol{\lambda}^{(t+1)}=\boldsymbol{\lambda}^{(t)}+\boldsymbol{\rho}^{(t)} \circ \nabla_\lambda\widehat{\mathcal{L}(\boldsymbol{\lambda}^{(t)})}$
\EndWhile
\end{algorithmic}
\end{algorithm}

\begin{algorithm}[H]
\setstretch{1.35}
\caption{\setstretch{1.35}\smallskip Subsampling VI for Bayesian SADR}\label{alg:subsampling}
\smallskip Set $t=0$. Initialize $\lambdavec^{(t)}=((\muvec^{(t)})^\top, (\text{vech}(B^{(t)}))^\top , (\dvec^{(t)})^\top)^\top$
\begin{algorithmic}[1]
\While{Stopping rule is not satisfied}
\State Set $t = t + 1$
	\State Generate $m=1,...,M$ samples:
$\quad \boldsymbol{\epsilon}_m^{(t)} \sim \mathcal{N}(\nullvec,I_{p_\bullet})$ and $\boldsymbol{\xi}_m^{(t)} \sim \mathcal{N}(\nullvec,I_{k})$
	\State  Compute $\boldsymbol{\theta}_m^{(t)}=t(\zetavec^{(t)},\lambdavec^{(t)})$
	\State  Sample $N_\text{sub}$ observations
	\State  Construct unbiased estimates 
 \begin{equation*}
 \nabla_\lambda \widehat{ \mathcal{L}(\boldsymbol{\lambda}^{(t)})}=\frac{1}{M}\sum_{m=1}^M \left( \frac{d (\thetavec^{(t)})^\top}{d\boldsymbol{\lambda}^{(t)}}\nabla_{\theta^{(t)}} \left[\log p(\thetavec^{(t)})\right.\right.\left.\left. +\frac{n}{n_\text{sub}}\sum_{i=1}^{n_{\text{sub}}}\log \sum p(y_i|\thetavec^{(t)}) - \log q_{\lambda^{(t)}}(\thetavec^{(t)})\right]\right)
 		\end{equation*}
\State Compute $\rhovec^{(t)}$ using ADADELTA
\State Set $\boldsymbol{\lambda}^{(t+1)}=\boldsymbol{\lambda}^{(t)}+\boldsymbol{\rho}^{(t)} \circ \nabla_\lambda\widehat{\mathcal{L}(\boldsymbol{\lambda}^{(t)})}$
 \EndWhile
\end{algorithmic}
\end{algorithm}

\begin{algorithm}[H]
\setstretch{1.35}
\caption{\setstretch{1.35}\smallskip Global annealing of VI for Bayesian SADR}\label{alg:anneal}
\smallskip Set $t=0$. Initialize $\lambdavec^{(t)}=((\muvec^{(t)})^\top, (\text{vech}(B^{(t)}))^\top , (\dvec^{(t)})^\top)^\top$ and $T_0$
\begin{algorithmic}[1]
\While{Stopping rule is not satisfied}
\State Set $t = t + 1$
	\State Generate $m=1,...,M$ samples:
$\quad \boldsymbol{\epsilon}_m^{(t)} \sim \mathcal{N}(\nullvec,I_{p_\bullet})$ and $\boldsymbol{\xi}_m^{(t)} \sim \mathcal{N}(\nullvec,I_{k})$
	\State  Compute $\boldsymbol{\theta}_m^{(t)}=t(\zetavec^{(t)},\lambdavec^{(t)})$
	\State  Construct unbiased estimates 
 \begin{equation*}
\nabla_\lambda \widehat{ \mathcal{L}(\boldsymbol{\lambda}^{(t)})}=\frac{1}{M}\sum_{m=1}^M \left( \frac{d (\thetavec^{(t)})^\top}{d\boldsymbol{\lambda}^{(t)}}\nabla_{\theta^{(t)}} \left[\log (p(\thetavec^{(t)})\right.\right.
 		{\left.\left. +\frac{1}{T}\sum_{i=1}^{n}\log \sum p(y_i|\thetavec^{(t)}) - \log q_{\lambda^{(t)}}(\thetavec^{(t)})\right]\right)}
 			\end{equation*}
\State Compute $\rhovec^{(t)}$ using ADADELTA
\State Set $\boldsymbol{\lambda}^{(t+1)}=\boldsymbol{\lambda}^{(t)}+\boldsymbol{\rho}^{(t)} \circ \nabla_\lambda\widehat{\mathcal{L}(\boldsymbol{\lambda}^{(t)})}$
	\State 	If $t\in \lbrace 100,200,300,\ldots,9000\rbrace$ reduce $T$ 
 \EndWhile
\end{algorithmic}
\end{algorithm}

\begin{algorithm}[H]
\setstretch{1.35}
\caption{\setstretch{1.35}\smallskip Robust fitting of VI for Bayesian SADR.}\label{alg:robust}
\smallskip Initialize with $t=0$ $\lambdavec_\theta^{(t)}=(\muvec_\theta^{(t)},\text{vec}(B^{(t)}),\dvec^{(t)}))$ , $\lambdavec_{\widetilde{w}}^{(t)}=(\muvec_{\widetilde{w}}^{(t)})=0.98\times \boldsymbol{1},\boldsymbol{\varrho}^{(t)}=\boldsymbol{1}$); $(\lambdavec^{(t)}=(\lambdavec_\theta^{(t)})^\top,(\lambdavec_{\widetilde{w}}^{(t)})^\top)^\top$:
\begin{algorithmic}[1]
\While{Stopping rule is not satisfied}
\State Set $t = t + 1$
\State Generate $m=1,...,M$ samples:  $\boldsymbol{\xi}_m^{(t)} \sim \mathcal{N}(\nullvec,I_k)$ and $\boldsymbol{\epsilon}_m^{\text{total}(t)} \sim \mathcal{N}(0,I_{p_\bullet+n})$
\State  Compute $\boldsymbol{\phi}_m^{(t)}=\muvec^{\text{total} (t)}+B^{\text{total} (t)}\boldsymbol{\zeta}_m^{(t)}+\dvec^{\text{total} (t)}\boldsymbol{\epsilon}_m^{\text{total} (t)}$
\State  Construct unbiased estimates 
	$\nabla_{\lambda_\theta} \widehat{ \mathcal{L}(\lambdavec^{(t)})}$ and $\nabla_{\lambda_{\widetilde{w}}} \widehat{ \mathcal{L}(\lambdavec^{(t)})}$
\State Compute $\rhovec^{(t)}$ using ADADELTA
\State Set $\boldsymbol{\lambda}^{(t+1)}=\boldsymbol{\lambda}^{(t)}+\boldsymbol{\rho}^{(t)} \circ \nabla_\lambda\widehat{\mathcal{L}(\boldsymbol{\lambda}^{(t)})}$
\EndWhile
\end{algorithmic}
\end{algorithm}

\clearpage

\section{Simulations}\label{app:simulation}
As described in the main text, we use four existing data sets to conduct the simulation study. All four data sets are available publicly in R packages listed in Table~\ref{tab:data sets}. Goal is to 
\subsection{Histograms of responses}
To evaluate the performance of the newly introduced estimation approach, we look at four rather different response distributions. The kernel densities of each response can be seen in Figure~\ref{fig:hist_y}. 
\begin{figure}[H]
		\includegraphics[width = 0.49\textwidth]{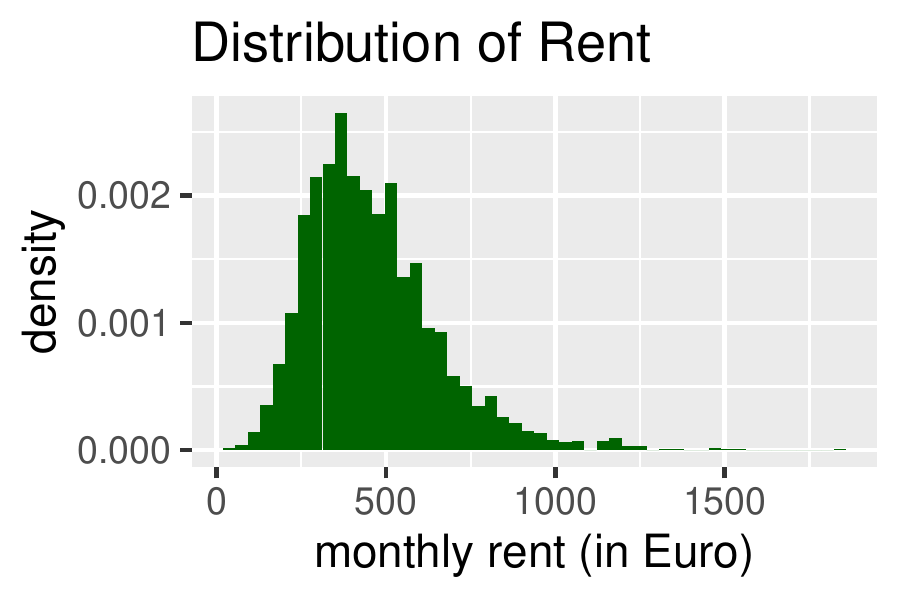}
		\includegraphics[width = 0.49\textwidth]{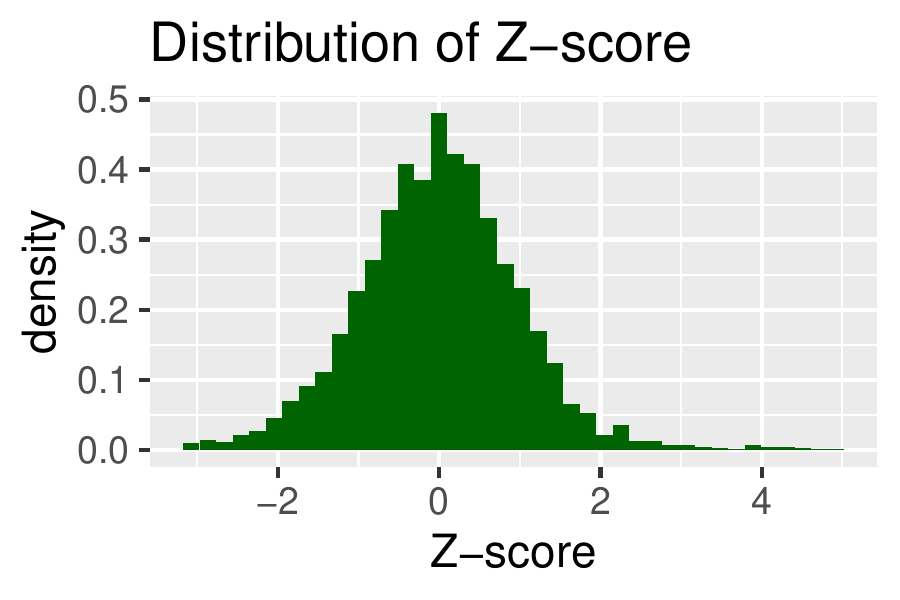}
		\includegraphics[width = 0.49\textwidth]{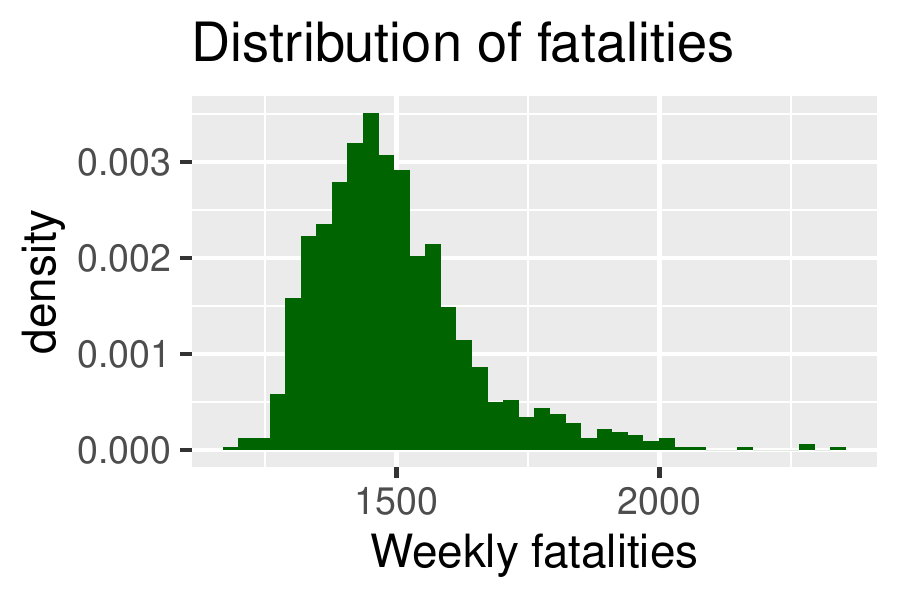}
		\includegraphics[width = 0.49\textwidth]{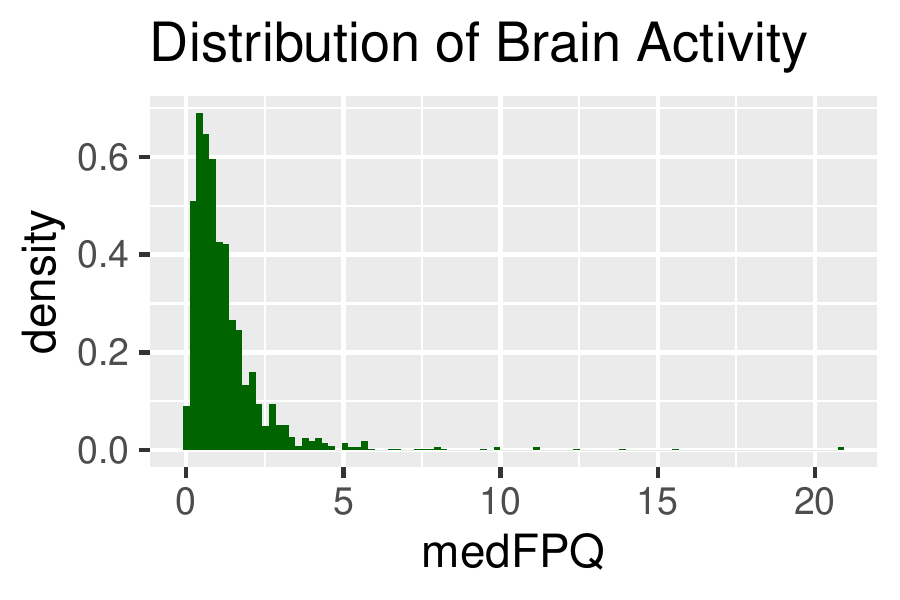}
 \caption{Histograms of the four responses.}\label{fig:hist_y}
\end{figure}

\subsection{Tuning \texttt{abamlss}}\label{app:defaults}
To identify the optimal parameters for our estimation approach, we conduct a simulations study where we compare different parameter combinations. Next to identifying how large the sampling for the estimation of the gradient has to be (\textit{M}), we want to analyze the effects of introducing subsampling and a Gibbs step. Additionally, we look at different prior distributions for $\tauvec$. Table~\ref{tab:abbr_methods} lists all settings we compare. While Figure~\ref{fig:abamlss_settings} depicts the simulation results. A summary of the findings can be found in Section~\ref{sec:simul}.
\begin{table}[H]\small{
	\begin{tabular}{l|l}
		\hline \hline
		\texttt{abamlss} & $M$=1, $n_{\mbox{\scriptsize{sub}}}=n$, IG prior \\
		&\\
		\texttt{abamlss with SD prior}  & $M$=1, $n_{\mbox{\scriptsize{sub}}}=n$, SD prior \\
		&\\
		\texttt{abamlss with $M$=5} & $M$=5, $n_{\mbox{\scriptsize{sub}}}=n$, IG prior \\
		&\\
		\texttt{abamlss with subsampling} & $M$=1, $n_{\mbox{\scriptsize{sub}}}=0.4\times n$, IG prior  \\
		&\\
		\texttt{abamlss with Gibbs step} & \texttt{abamlss} settings with Gibbs step extension\\
		&\\
		\texttt{abamlss with Gibbs step \& $M$=5} & \texttt{abamlss} with $M=5$ setting with Gibbs step \\
		\hline
	\end{tabular}}\caption[Comparing different \texttt{abamlss} specifications]{\small{Reported are the different \texttt{abamlss} settings used for the size of the Monte Carlo samples $M$, the number of data points evaluated in each iteration $n_{\mbox{\scriptsize{sub}}}$, the hyperprior for the smoothing variance and whether an extension of Section \ref{sec:extend} is used. }}\label{tab:abbr_methods}
\end{table}
\begin{figure}[H]
 \centering
 \caption*{Tuning \texttt{abamlss}.}
 \includegraphics[scale = 0.35]{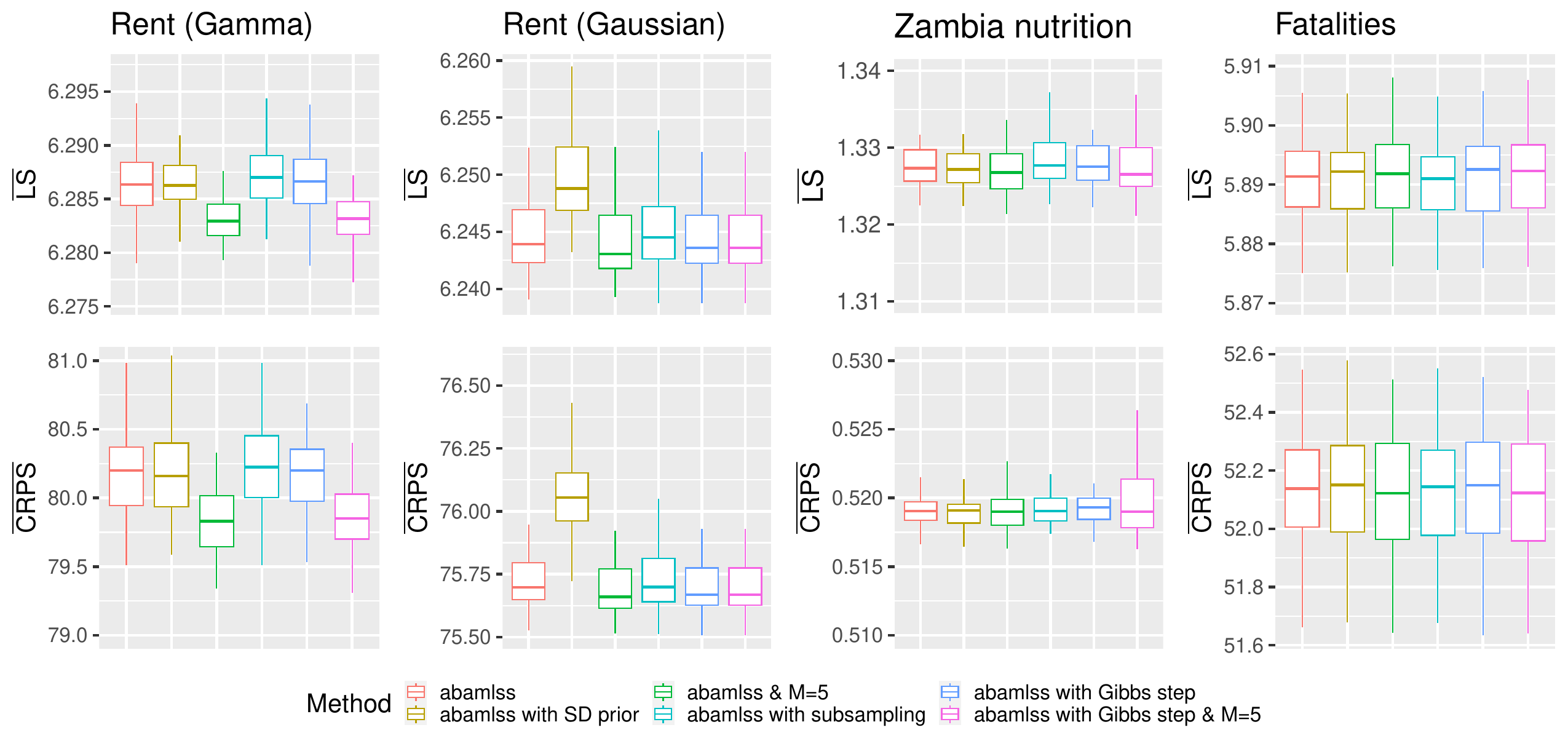}
 \caption[Simulation results]{\small{For each data set, the mean LS and mean CRPS are depicted for 40 replications for different settings of abamlss. \texttt{abamlss} stands for the general estimation procedure with one MC sample $M$=1 to estimate the gradient, the inverse gamma prior for the scaling parameters $\tauvec^2$ and with no subsampling and no inclusion of a Gibbs step. }}\label{fig:abamlss_settings}
  \end{figure}
The global annealing approach introduces a dynamic global temperature that gradually decreases over time. This temperature manipulation strategy serves a dual purpose: it initially encourages exploration of the optimization function, and as it decreases, it shifts focus towards fine-tuning the fit. To rigorously assess the impact of this temperature variation, we undertake a dedicated simulation study. In this study, we aim not only to evaluate the effectiveness of global annealing but also to pinpoint the optimal starting temperature, denoted as $T_0$. The results are visually presented in Figure~\ref{fig:ann}, and upon careful analysis, we draw the conclusion that there isn't a significant performance improvement.
\begin{figure}[H]
\caption*{Global annealing.}
\centering\includegraphics[scale=0.25]{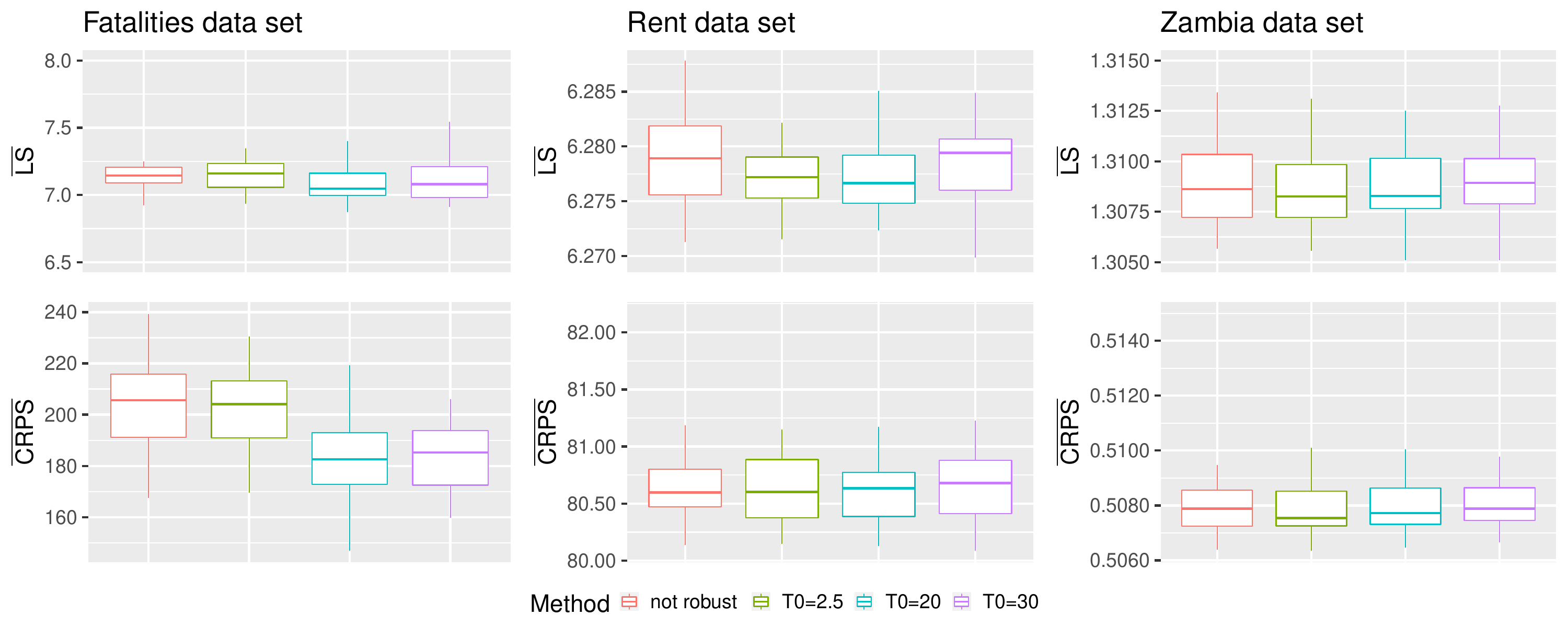}\caption[Annealing for different data sets]{\small{Based on \texttt{abamlss} we simulated  41 replicates where 40 were used to estimate the model with the non robust version \texttt{abamlss} and with three different starting temperatures of AVI. The last replicate was used for out of sample prediction. The mean CRPS and mean LS of each replicate is plotted}}\label{fig:ann}
\end{figure}
 \subsection{Robust fitting}\label{app:simrobust}
For the robust version of the estimation procedure introduced in Section~\ref{sec:extend}, we need to identify the optimal hyperparameters for the prior distributions of the weights introduced. Important to note is that this depends on the specific data set at hand so that the results found here cannot be necessarily transfered to other data. Figure~\ref{fig:brain_hyper} presents the results of the simulation. For the normal data set we see that \texttt{abamlss} outperforms any robust method when looking at the mean LS and mean CRPS. For the contaminated data set however, all robust methods outperform \texttt{abamlss}. The robust method with a = 0.2 and b = 0.01 shows best prediction accuracy so that we will choose this distribution as our hyperprior specification. The Jeffreys prior (a = 0.5 and b = 0.5) of the beta distribution is the second best specification. 
\begin{figure}[H]
	\centering
 \caption*{Robust fitting.}
	\includegraphics[scale=0.3]{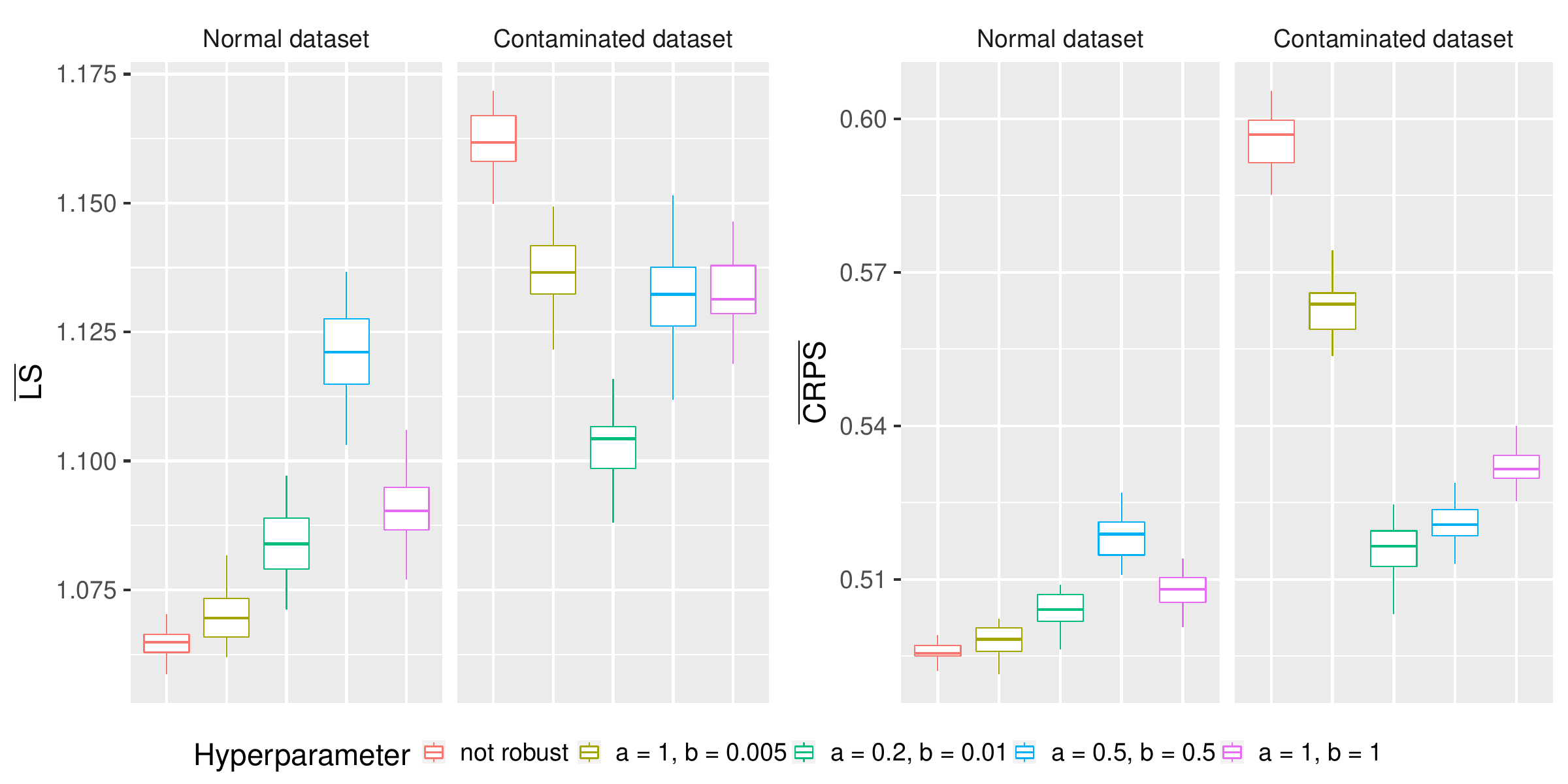}\caption[Estimation performance for varying hyperprior specifications]{\small{We simulated 40 replicates based on the DGP of an \texttt{abamlss} estimation and generated for each replicate a contaminated data set. We estimated the basic \texttt{abamlss} algorithm next to different versions of \texttt{robust abamlss} algorithm where we vary the specification of the hyperparameter of the beta distribution. We then evaluated the prediction accuracy of the 40 normal replicates and the 40 contaminated replicates on the same 41. replicate and calculated the mean CRPS and mean LS}.}\label{fig:brain_hyper}
\end{figure}

 \begin{table}[H]
	\small{
		\centering	\begin{tabular}{l|cc|cc}
  \hline\hline
			 &  \multicolumn{2}{c}{Normal data set (min)}  &  \multicolumn{2}{c}{Contaminated data set (min)}\\
		Method	 & not robust & robust & not robust & robust \\
			\hline 
			\texttt{bamlss} & 1.50 &- &1.49& -\\
			\texttt{abamlss} &2.50  & 12.90&1.42&12.90\\
			\texttt{gamlss} &0.10&13.77&1.44& 20.66 \\  
			\hline\hline
		\end{tabular}\caption[Computation time (min) - Bayesian data reweighting]{\small{We took the mean computation time in mins of 10 replicates for each estimation setting for the normal and the contaminated data set.  All settings used initialization and the stopping criteria based on the median lower bound described in Section \ref{sec:VB}}}\label{tab:times2}}
\end{table}
\section{Real Data Illustrations}\label{app:emp}
\subsection{Modelling infectious outbreaks}
We rely on the COVID-19 dataset provided by the Robert-Koch-Institute. We utilize the conveniently aggregated data, thoughtfully prepared by \citet{SchNicKau2021}, which is hosted at the Leibniz-Rechenzentrum (LRZ). The dataset can be accessed and downloaded from the following location: \url{https://syncandshare.lrz.de/getlink/fiPvZjnVzKNsuwZ5Upzgy7/}. 

Figure~\ref{fig:covid} depicts the histogram of the response while Table~\ref{tab:covid_descriptives} present additional information of its distribution.
\begin{figure}[H]
    \centering
    \includegraphics[scale = 0.7]{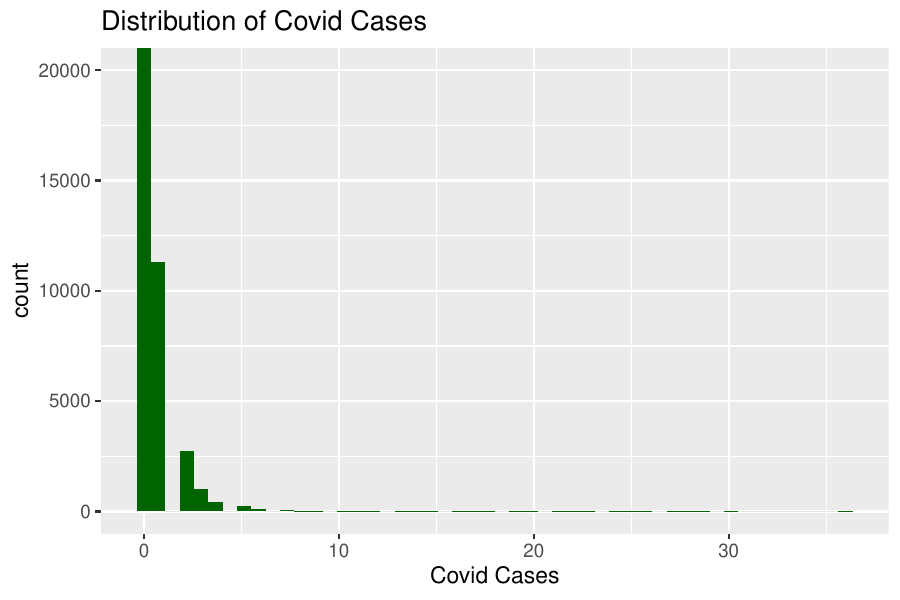}
    \caption{Histogram of the COVID cases. Important to note is that the first bar is cut for better visibility of the rest of the histogram.}\label{fig:covid}
\end{figure}

\begin{table}[H]
    \centering
    \begin{tabular}{c|c|c|c}
        min   & median & mean & max \\
        \hline \hline
        0 & 0 & 0.04 & 36 
    \end{tabular}
    \caption{Descriptives of COVID cases in the time range between 2020-08-25 up until 2020-09-14}
    \label{tab:covid_descriptives}
\end{table}
\end{document}